\documentclass[runningheads,envcountsame]{llncs}

\usepackage{microtype,xspace}
\usepackage{url}
\usepackage[T1]{fontenc}
\usepackage{bbm,mathtools,amssymb,stmaryrd}
\usepackage{bbold}
\usepackage{paralist}
\usepackage{tikz}
\usetikzlibrary{decorations}
\usetikzlibrary{automata}
\usetikzlibrary{shapes.multipart}
\usetikzlibrary{arrows}

\usepackage{mysty}

\title{Structural Refinement for the Modal nu-Calculus}

\author{Uli Fahrenberg \and Axel Legay \and Louis-Marie Traonouez}
\institute{Inria / IRISA, Campus de Beaulieu, 35042 Rennes CEDEX,
  France}

\begin{document}

\maketitle

\begin{abstract}
  We introduce a new notion of structural refinement, a sound
  abstraction of logical implication, for the modal nu-calculus.  Using
  new translations between the modal nu-calculus and disjunctive modal
  transition systems, we show that these two specification formalisms
  are structurally equivalent.

  Using our translations, we also transfer the structural operations of
  composition and quotient from disjunctive modal transition systems to
  the modal nu-calculus.  This shows that the modal nu-calculus supports
  composition and decomposition of specifications.
\end{abstract}

\section{Introduction}

There are two conceptually different approaches for the specification
and verification of properties of formal models.  \emph{Logical}
approaches make use of logical formulae for expressing properties and
then rely on efficient model checking algorithms for verifying whether
or not a model satisfies a formula.  \emph{Automata}-based approaches,
on the other hand, exploit equivalence or refinement checking for
verifying properties, given that models and properties are specified
using the same (or a closely related) formalism.

The logical approaches have been quite successful, with a plethora of
logical formalisms available and a number of successful model checking
tools.  One particularly interesting such formalism is the modal
$\mu$-calculus~\cite{DBLP:journals/tcs/Kozen83}, which is universal in
the sense that it generalizes most other temporal logics, yet
mathematically simple and amenable to analysis.

One central problem in the verification of formal properties is
\emph{state space explosion}: when a model is composed of many
components, the state space of the combined system quickly grows too big
to be analyzed.  To combat this problem, one approach is to employ
\emph{compositionality}.  When a model consists of several components,
each component would be model checked by itself, and then the
components' properties would be composed to yield a property which
automatically is satisfied by the combined model.

Similarly, given a global property of a model and a component of the
model that is already known to satisfy a local property, one would be
able to \emph{decompose} automatically, from the global property and the
local property, a new property which the rest of the model must satisfy.
We refer to~\cite{DBLP:conf/concur/Larsen90} for a good account of these
and other features which one would wish specifications to have.

As an alternative to logical specification formalisms and with an eye to
compositionality and decomposition, automata-based \emph{behavioral}
specifications were introduced in~\cite{DBLP:conf/avmfss/Larsen89}.
Here the specification formalism is a generalization of the modeling
formalism, and the satisfaction relation between models and
specifications is generalized to a refinement relation between
specifications, which resembles simulation and bisimulation and can be
checked with similar algorithms.

For an example, we refer to Fig.~\ref{fi:example} which shows the
property informally specified as ``after a \req(uest), no \idle(ing) is
allowed, but only \work, until \grant\ is executed'' using the logical
formalisms of CTL~\cite{DBLP:conf/lop/ClarkeE81} and the modal
$\mu$-calculus~\cite{DBLP:journals/tcs/Kozen83} and the behavioral
formalism of disjunctive modal transition
systems~\cite{DBLP:conf/lics/LarsenX90}.

\begin{figure}[tp]
  \begin{minipage}[c]{7.5cm}
  \begin{gather*}
    \text{AG}(\req\Rightarrow \text{AX}(\work \text{ AW } \grant))
    \\[3ex]
    \begin{aligned}
      \nu X.\big( &[ \grant, \idle, \work] X\land \\
      &\quad [ \req] \nu Y.( \langle \work\rangle Y\lor \langle
      \grant\rangle X)\land[ \idle, \req] \lff\big)
    \end{aligned} \\[3ex]
    \begin{aligned}
      X &\overset \nu= [ \grant, \idle, \work] X\land[ \req] Y \\
      Y &\overset \nu= ( \langle \work\rangle Y\lor \langle \grant\rangle
      X)\land[ \idle, \req] \lff
    \end{aligned}
  \end{gather*}
  \end{minipage}
  \begin{minipage}[c]{4cm}
    \begin{tikzpicture}[->, >=stealth', font=\footnotesize,
      state/.style={shape=circle, draw, initial text=,inner
        sep=.5mm,minimum size=2mm}, scale=1.4]
      \node[state,initial] (X) {};
      \node[state] (Y) at (2,0) {};
      \path[densely dashed,->] (X) edge node[above]{$\req$} (Y); 
      \path[densely dashed,->] (X) edge[loop above]
      node[above]{$\grant,\work,\idle$} (X);
      \coordinate (Yd) at (1.1,-.4) {};
      \path (Y) edge[-, thick] (Yd);
      \path (Yd) edge[out= 180, in= -45] node[below]{\grant} (X);
      \path (Yd) edge[out= -45, in= 240] node[below]{\work} (Y);
    \end{tikzpicture}
  \end{minipage}
  \caption{%
    \label{fi:example}
    An example property specified in CTL (top left), in the modal
    $\mu$-calculus (below left), as a modal equation system (third
    left), and as a DMTS (right).}
\end{figure}
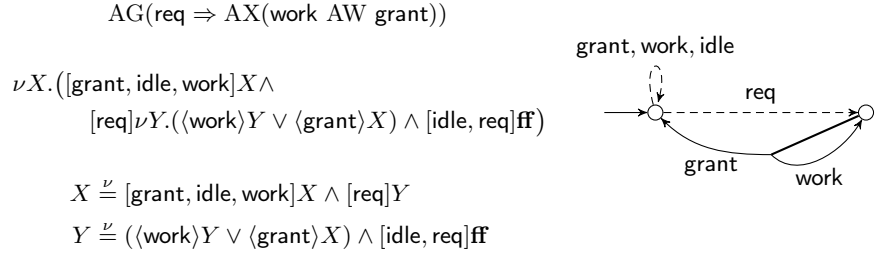

The precise relationship between logical and behavioral specification
formalisms has been subject to some investigation.
In~\cite{DBLP:conf/avmfss/Larsen89}, Larsen shows that any \emph{modal
  transition system} can be translated to a formula in Hennessy-Milner
logic which is equivalent in the sense of admitting the same models.
Conversely, Boudol and Larsen show in~\cite{DBLP:journals/tcs/BoudolL92}
that any formula in Hennessy-Milner logic is equivalent to a finite
disjunction of modal transition systems.

We have picked up this work in~\cite{DBLP:conf/concur/BenesDFKL13},
where we show that any \emph{disjunctive modal transition system} (DMTS)
is equivalent to a formula in the \emph{modal $\nu$-calculus}, the
safety fragment of the modal $\mu$-calculus which uses only maximal
fixed points, and vice versa.  (Note that the modal $\nu$-calculus is
equivalent to Hennessy-Milner logic with recursion and maximal fixed
points.)  Moreover, we show in~\cite{DBLP:conf/concur/BenesDFKL13} that
DMTS are as expressive as (non-deterministic) \emph{acceptance
  automata}~\cite{report/irisa/Raclet07,DBLP:journals/entcs/Raclet08}.
Together with the inclusions of~\cite{DBLP:conf/atva/BenesKLMS11}, this
settles the expressivity question for behavioral specifications: they
are at most as expressive as the modal $\nu$-calculus.

\medskip

In this paper, we show that not only are DMTS as expressive as the modal
$\nu$-calculus, but the two formalisms are \emph{structurally
  equivalent}.  Introducing a new notion of structural refinement for
the modal $\nu$-calculus (a sound abstraction of logical implication),
we show that one can freely translate between the modal $\nu$-calculus
and DMTS, while preserving structural refinement.

DMTS form a \emph{complete specification
  theory}~\cite{DBLP:conf/fase/BauerDHLLNW12} in that they both admit
logical operations of conjunction and disjunction and structural
operations of composition and
quotient~\cite{DBLP:conf/concur/BenesDFKL13}.  Hence they support full
compositionality and decomposition in the sense
of~\cite{DBLP:conf/concur/Larsen90}.  Using our translations, we can
transport these notions to the modal $\nu$-calculus, thus also turning
the modal $\nu$-calculus into a complete specification theory.

In order to arrive at our translations, we first recall DMTS and
(non-determi\-nistic) \emph{acceptance automata} in
Section~\ref{se:structspec}.  We also introduce a new hybrid modal
logic, which can serve as compact representation for acceptance automata
and should be of interest in itself.  Afterwards we show, using the
translations introduced in~\cite{DBLP:conf/concur/BenesDFKL13}, that
these formalisms are structurally equivalent.

In Section~\ref{se:mu} we recall the modal $\nu$-calculus and review the
translations between DMTS and the modal $\nu$-calculus which were
introduced in~\cite{DBLP:conf/concur/BenesDFKL13}.  These in turn are
based on work by Boudol and Larsen in~\cite{DBLP:journals/tcs/BoudolL92,DBLP:conf/avmfss/Larsen89}, hence fairly standard.  We show that,
though semantically correct, the two translations are structurally
\emph{mismatched} in that they relate DMTS refinement to two different
notions of $\nu$-calculus refinement.  To fix the mismatch, we introduce
a new translation from the modal $\nu$-calculus to DMTS and show that
using this translation, the two formalisms are structurally equivalent.

In Section~\ref{se:spectheory}, we use our translations to turn the
modal $\nu$-calculus into a complete specification theory.  We remark
that all our translations and constructions are based on a new
\emph{normal form} for $\nu$-calculus expressions, and that turning a
$\nu$-calculus expression into normal form may incur an exponential
blow-up.  However, the translations and constructions preserve the
normal form, so that this translation only need be applied once in the
beginning.

We also note that composition and quotient operators are used in other
logics such as \eg~spatial~\cite{DBLP:journals/iandc/CairesC03} or
separation
logics~\cite{DBLP:conf/lics/Reynolds02,DBLP:conf/csl/OHearnRY01}.
However, in these logics they are treated as \emph{first-class}
operators, \ie~as part of the formal syntax.  In our approach, on the
other hand, they are defined as operations on logical expressions which
as results again yield logical expressions (without compositions or
quotients).

Note that some proofs have been relegated to a separate appendix.

\section{Structural Specification Formalisms}
\label{se:structspec}

Let $\Sigma$ be a finite set of labels.  A \emph{labeled transition
  system} (LTS) is a structure $\cI=( S, S^0, \omust)$ consisting of a
finite set of \emph{states} $S$, a subset $S^0\subseteq S$ of
\emph{initial states} and a \emph{transition relation} $\omust\subseteq
S\times \Sigma\times S$.

\subsection{Disjunctive modal transition systems}

A \emph{disjunctive modal transition system} (DMTS) is a structure
$\cD=( S, S^0, \omay, \omust)$ consisting of finite sets $S\supseteq
S^0$ of states and initial states, a \emph{may}-transition relation
$\omay\subseteq S\times \Sigma\times S$, and a \emph{disjunctive
  must}-transition relation $\omust\subseteq S\times 2^{ \Sigma\times
  S}$.  It is assumed that for all $( s, N)\in \omust$ and all $( a,
t)\in N$, $( s, a, t)\in \omay$.

As customary, we write $s\may a t$ instead of $( s, a, t)\in \omay$,
$s\must{} N$ instead of $( s, N)\in \omust$, $s\may a$ if there exists
$t$ for which $s\may a t$, and $s\notmay a$ if there does not.  

The intuition is that may-transitions $s\may a t$ specify which
transitions are permitted in an implementation, whereas a
must-transitions $s\must{} N$ stipulates a disjunctive requirement: at
least one of the choices $( a, t)\in N$ must be implemented.  A DMTS $(
S, S^0, \omay, \omust)$ is an \emph{implementation} if $\omust=\{( s,\{(
a, t)\})\mid s\may a t\}$; DMTS implementations are precisely LTS.

DMTS were introduced in~\cite{DBLP:conf/lics/LarsenX90} in the context
of equation solving, or \emph{quotient}, for specifications and are used
\eg~in \cite{DBLP:conf/atva/BenesCK11} for LTL model checking.  They are
a natural closure of \emph{modal transition systems}
(MTS)~\cite{DBLP:conf/avmfss/Larsen89} in which all disjunctive
must-transitions $s\must{} N$ lead to singletons $N=\{( a, t)\}$.

Let $\cD_1=( S_1, S^0_1, \omay_1, \omust_1)$, $\cD_2=( S_2, S^0_2,
\omay_2, \omust_2)$ be DMTS.  A relation $R\subseteq S_1\times S_2$ is a
\emph{modal refinement} if it holds for all $( s_1, s_2)\in R$ that
\begin{compactitem}
\item for all $s_1\may a t_1$ there is $t_2\in S_2$ with $s_2\may a t_2$
  and $( t_1, t_2)\in R$, and
\item for all $s_2\must{} N_2$ there is $s_1\must{} N_1$ such that for
  each $( a, t_1)\in N_1$ there is $( a, t_2)\in N_2$ with $( t_1,
  t_2)\in R$.
\end{compactitem}
We say that $\cD_1$ \emph{modally refines} $\cD_2$, denoted $\cD_1\mr
\cD_2$, whenever there exists a modal refinement $R$ such that for all
$s^0_1\in S^0_1$, there exists $s^0_2\in S^0_2$ for which $( s^0_1,
s^0_2)\in R$.  We write $\cD_1\mreq \cD_2$ if $\cD_1\mr \cD_2$ and
$\cD_2\mr \cD_1$.  For states $s_1\in S_1$, $s_2\in S_2$, we write
$s_1\mr s_2$ if the DMTS $( S_1,\{ s_1\}, \omay_1, \omust_1)\mr( S_2,\{
s_2\}, \omay_2, \omust_2)$.

Note that modal refinement is reflexive and transitive, \ie~a preorder
on DMTS.  Also, the relation on states $\mathord\mr\subseteq S_1\times
S_2$ defined above is itself a modal refinement, indeed the maximal
modal refinement under the subset ordering.

The \emph{set of implementations} of an DMTS $\cD$ is $\sem \cD=\{
\cI\mr \cD\mid \cI~\text{implement-}\linebreak \text{ation}\}$.  This
is, thus, the set of all LTS which satisfy the specification given by
the DMTS $\cD$.  We say that $\cD_1$ \emph{thoroughly refines} $\cD_2$,
and write $\cD_1\tr \cD_2$, if $\sem{ \cD_1}\subseteq \sem{ \cD_2}$.  We
write $\cD_1\treq \cD_2$ if $\cD_1\tr \cD_2$ and $\cD_2\tr \cD_1$.  For
states $s_1\in S_1$, $s_2\in S_2$, we write $\sem{ s_1}= \sem{( S_1,\{
  s_1\}, \omay_1, \omust_1)}$ and $s_1\tr s_2$ if $\sem{ s_1}\subseteq
\sem{ s_2}$.

The below proposition, which follows directly from transitivity of modal
refinement, shows that modal refinement is \emph{sound} with respect to
thorough refinement; in the context of specification theories, this is
what one would expect, and we only include it for completeness of
presentation.  It can be shown that modal refinement is also
\emph{complete} for \emph{deterministic}
DMTS~\cite{DBLP:journals/tcs/BenesKLS09}, but we will not need this
here.

\begin{proposition}
  \label{pr:mrvstr}
  For all DMTS $\cD_1$, $\cD_2$, $\cD_1\mr \cD_2$ implies $\cD_1\tr
  \cD_2$. \noproof
\end{proposition}

We introduce a new construction on DMTS which will be of interest for
us; intuitively, it adds all possible may-transitions without changing
the implementation semantics.  The \emph{may-completion} of a DMTS
$\cD=( S, S^0, \omay, \omust)$ is $\maycomp{ \cD}=( S, S^0,
\omay_\omaycomp, \omust)$ with
\begin{equation*}
  \omay_{ \omaycomp}=\{( s, a, t')\subseteq S\times \Sigma\times S\mid
  \exists( s, a, t)\in \omay: t'\tr t\}.
\end{equation*}
Note that to compute the may-completion of a DMTS, one has to decide
thorough refinements, hence this computation (or, more precisely,
deciding whether a given DMTS is may-complete) is
EXPTIME-complete~\cite{DBLP:journals/iandc/BenesKLS12}.  We show an
example of a may-completion in Fig.~\ref{fi:may-complete}.

\begin{figure}[tp]
  \centering
  \begin{tikzpicture}[->, >=stealth', font=\footnotesize,
    state/.style={shape=circle, draw, initial text=,inner
      sep=.5mm,minimum size=5mm}, scale=.75]
    \begin{scope}
      \node at (0,2.5) {$\cD$};
      \node[state, initial] (s') at (.5,0) {$s\vphantom'$};
      \node[state] (t') at (2,1.5) {$t_1\vphantom'$};
      \node[state] (t''') at (2,-1) {$t_3\vphantom'$};
      \node[state] (u') at (4,2.5) {$u_1\vphantom'$};
      \node[state] (u'') at (4,.5) {$u_2\vphantom'$};
      \node[state] (u''') at (4,-1) {$u_3\vphantom'$};
      \node[state] (v') at (6,2.5) {$v_1\vphantom'$};
      \node[state] (v''') at (6,-1) {$v_3\vphantom'$};
      \path[densely dashed] (s') edge node[above] {$a$} (t');
      \path[densely dashed] (t') edge node[above] {$a$} (u');
      \path[densely dashed] (t') edge node[below] {$a$} (u'');
      \path[densely dashed] (t''') edge node[below] {$a$} (u''');
      \path (u') edge node[above] {$a$} (v');
      \path[densely dashed, loop right] (u'') edge node[below] {$d$} (u'');
      \path[densely dashed] (u''') edge node[below] {$a$} (v''');
      \path[densely dashed, loop right] (v') edge node[right] {$b,c$}
      (v');
      \path[densely dashed, loop right] (v''') edge node[right] {$b$}
      (v''');
    \end{scope}
    \begin{scope}[xshift=25em]
      \node at (0.5,2.5) {$\maycomp{\cD}$};
      \node[state, initial] (s') at (.5,0) {$s'$};
      \node[state] (t') at (2,1.5) {$t_1'$};
      \node[state] (t''') at (2,-1) {$t_3'$};
      \node[state] (u') at (4,2.5) {$u_1'$};
      \node[state] (u'') at (4,.5) {$u_2'$};
      \node[state] (u''') at (4,-1) {$u_3'$};
      \node[state] (v') at (6,2.5) {$v_1'$};
      \node[state] (v''') at (6,-1) {$v_3'$};
      \path[densely dashed] (s') edge node[above] {$a$} (t');
      \path[densely dashed] (t') edge node[above] {$a$} (u');
      \path[densely dashed] (t') edge node[below] {$a$} (u'');
      \path[densely dashed] (t''') edge node[below] {$a$} (u''');
      \path (u') edge node[above] {$a$} (v');
      \path[densely dashed] (u''') edge node[below] {$a$} (v''');
      \path[densely dashed, loop right] (v') edge node[right] {$b, c$} (v');
      \path[densely dashed, loop right] (v''') edge node[right] {$b$} (v''');
      \path[densely dashed, loop right] (u'') edge node[right] {$d$} (u'');
      \path[densely dashed] (s') edge node[below] {$a$} (t''');
      \path[densely dashed] (u') edge[out=320, in=110] node[right,
      pos=.35] {$a$} (v''');
      \path[densely dashed] (v') edge[out=290, in=70] node[right] {$b,
        c$} (v''');
      \path[dotted] (t''') edge node[rotate=90] {$\tr$} (t');
      \path[dotted] (v''') edge node[rotate=90] {$\tr$} (v');
    \end{scope}
  \end{tikzpicture}
  \caption{%
    \label{fi:may-complete}
    A MTS $\cD$ (left) and its may-completion $\maycomp{ \cD}$ (right).
    In $\maycomp \cD$, the semantic inclusions which lead to extra
    may-transitions are depicted with dotted arrows.}
\end{figure}
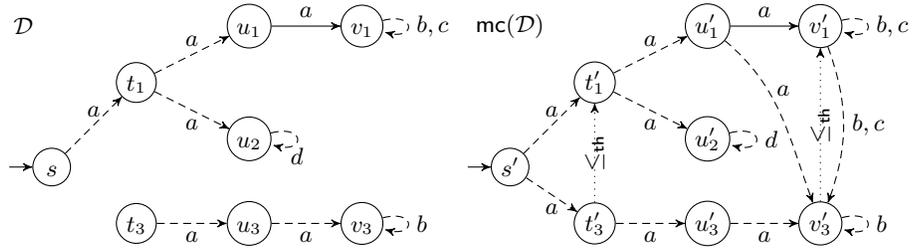

\begin{proposition}
  \label{pr:maycomp}
  For any DMTS $\cD$, $\cD\mr \maycomp{ \cD}$ and $\cD\treq \maycomp{
    \cD}$.
\end{proposition}

\begin{proof}
  It is always the case that $\cD \mr \cD$, and adding may transitions
  on the right side preserves modal refinement.  Therefore it is
  immediate that $\cD \mr \maycomp{\cD}$, hence also $\cD \tr
  \maycomp{\cD}$.

  To prove that $\maycomp{\cD} \tr \cD$, we consider an implementation
  $\cI \mr \maycomp{\cD}$; we must prove that $\cI \mr \cD$.  Write
  $\cD=( S, S^0, \omay, \omust)$, $\cI=( I, I^0, \omay_I, \omust_I)$ and
  $\maycomp{ \cD}=( S, S^0, \omay_\omaycomp, \omust)$.  Let $R
  \subseteq I \times S$ be the largest modal refinement between $\cI$
  and $\maycomp{ \cD}$.  We now prove that $R$ is also a modal
  refinement between $\cI$ and $\cD$.
  For all $(i,d) \in R$:
  \begin{compactitem}
  \item For all $i \may{a}_I i'$, there exists $d' \in S$ such that $d
    \may{a}_\omaycomp d'$ and $(i',d') \in R$.  Then by definition of
    $\omay_\omaycomp$, there exists $d'' \in S$ such that $d \may{a}
    d''$ and $\sem{d'} \subseteq \sem{d''}$.  $(i',d') \in R$ implies
    $i' \in \sem{d'}$, which implies $i' \in \sem{d''}$. This means that
    $i'\mr d''$, and since $R$ is the largest refinement relation in $I
    \times S$ it must be the case that $(i',d'') \in R$.
  \item The case of must transitions follows immediately, since must
    transitions are exactly the same in $\cD$ and $\maycomp{\cD}$. \qed
  \end{compactitem}
\end{proof}

\begin{example}
  \label{ex:mc_prop}
  The example in Fig.~\ref{fi:may-complete} shows that generally,
  $\maycomp{ \cD}\not\mr \cD$.  First, $t_3 \tr t_1$: For an
  implementation $\cI=( I, I^0, \must{})\in \sem{ t_3}$ with modal
  refinement $R\subseteq I\times\{ t_3, u_3, v_3\}$, define $R'\subseteq
  I\times\{ t_1, u_1, u_2, v_1\}$ by
  \begin{align*}
    R'=\{( i, t_1)\mid( i, t_3)\in R\} &\cup \{( i, v_1)\mid( i, v_3)\in
    R\} \\
    & \cup\{( i, u_1)\mid( i, u_3)\in R, i\must a\} \\
    & \cup\{( i, u_2)\mid( i, u_3)\in R, i\notmust a\},
  \end{align*}
  then $R'$ is a modal refinement $\cI\mr t_1$.  Similarly, $t_3'\tr
  t_1'$ in $\maycomp \cD$.

  On the other hand, $t_3\not\mr t_1$ (and similarly, $t_3'\not\mr
  t_1'$), because neither $u_3 \mr u_1$ nor $u_3 \mr u_2$.  Now in the
  modal refinement game between $\maycomp{ \cD}$ and $\cD$, the
  may-transition $s'\may a t_3'$ has to be matched by $s\may a t_1$, but
  then $t_3'\not\mr t_1$, hence $\maycomp{ \cD}\not\mr \cD$.

  Also, the may-completion does not necessarily preserve modal
  refinement: Consider the DMTS $\cD$ from Fig.~\ref{fi:may-complete}
  and $\cD_1$ from Fig.~\ref{fi:may-path}, and note first that
  $\maycomp{ \cD_1}= \cD_1$.  It is easy to see that $\cD \mr \cD_1$
  (just match states in $\cD$ with their double-prime cousins in
  $\cD_1$), but $\maycomp{\cD}\not\mr \maycomp{ \cD_1}= \cD_1$: the
  may-transition $s'\may a t_3'$ has to be matched by $s''\may a t_1''$
  and $t_3'\not\mr t_1''$.

  Lastly, the may-completion can also create modal refinement:
  Considering the DMTS $\cD_2$ from Fig.~\ref{fi:may-path}, we see that
  $\cD_2 \not\mr \cD$, but $\maycomp{ \cD_2}= \cD_2 \mr \maycomp{
    \cD}$. 
\end{example}

\begin{figure}[tp]
  \centering
  \begin{tikzpicture}[->, >=stealth', font=\footnotesize,
    state/.style={shape=circle, draw, initial text=,inner
      sep=.5mm,minimum size=5mm}, scale=.7] 
    \begin{scope}
      \node at (0,1.5) {$\cD_1$};
      \node[state, initial] (s') at (0,0) {$s''$};
      \node[state] (t') at (2,0) {$t_1''$};
      \node[state] (u') at (4,1) {$u_1''$};
      \node[state] (u'') at (4,-1) {$u_2''$};
      \node[state] (v') at (6,1) {$v_1''$};
      \path[densely dashed] (s') edge node[above] {$a$} (t');
      \path[densely dashed] (t') edge node[above] {$a$} (u');
      \path[densely dashed] (t') edge node[below] {$a$} (u'');
      \path (u') edge node[above] {$a$} (v');
      \path[densely dashed, loop right] (v') edge node[right] {$b, c$} (v');
      \path[densely dashed, loop right] (u'') edge node[right] {$d$} (u'');
    \end{scope}
    \begin{scope}[xshift=27em] 
      \node at (0,1.5) {$\cD_2$};
      \node[state, initial] (s) at (0,0) {$s''$};
      \node[state] (t) at (2,0) {$t''$};
      \node[state] (u) at (4,0) {$u''$};
      \node[state] (v) at (6,0) {$v''$};
      \path[densely dashed] (s) edge node[above] {$a$} (t);
      \path[densely dashed] (t) edge node[above] {$a$} (u);
      \path[densely dashed] (u) edge node[above] {$a$} (v);
      \path[densely dashed, loop right] (v) edge node[right] {$b$} (v);
    \end{scope}
  \end{tikzpicture}
  \caption{%
    \label{fi:may-path}
    DMTS $\cD_1$, $\cD_2$ from Example~\ref{ex:mc_prop}.
  }
\end{figure}
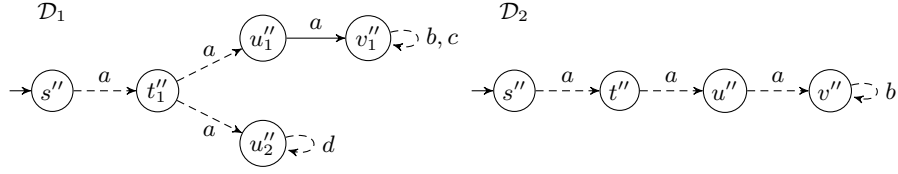

\subsection{Acceptance automata}

A (non-deterministic) \emph{acceptance automaton} (AA) is a structure
$\cA=( S, S^0, \Tran)$, with $S\supseteq S^0$ finite sets of states and
initial states and $\Tran: S\to 2^{ 2^{ \Sigma\times S}}$ an assignment
of \emph{transition constraints}.  We assume that for all $s^0\in S^0$,
$\Tran( s^0)\ne \emptyset$.

An AA is an \emph{implementation} if it holds for all $s\in S$ that
$\Tran( s)=\{ M\}$ is a singleton; hence also AA implementations are
precisely LTS.
Acceptance automata were first introduced
in~\cite{report/irisa/Raclet07} (see
also~\cite{DBLP:journals/entcs/Raclet08}, where a slightly different
language-based approach is taken), based on the notion of acceptance
trees in~\cite{DBLP:journals/jacm/Hennessy85}; however, there they are
restricted to be \emph{deterministic}.  We employ no such restriction
here.  The following notion of modal refinement for AA was also
introduced in~\cite{report/irisa/Raclet07}.

Let $\cA_1=( S_1, S^0_1, \Tran_1)$ and $\cA_2=( S_2, S^0_2, \Tran_2)$ be
AA.  A relation $R\subseteq S_1\times S_2$ is a \emph{modal refinement}
if it holds for all $( s_1, s_2)\in R$ and all $M_1\in \Tran_1( s_1)$
that there exists $M_2\in \Tran_2( s_2)$ such that
\begin{compactitem}
\item $\forall( a, t_1)\in M_1: \exists( a, t_2)\in M_2:( t_1, t_2)\in
  R$,
\item $\forall( a, t_2)\in M_2: \exists( a, t_1)\in M_1:( t_1, t_2)\in
  R$.
\end{compactitem}
As for DMTS, we write $\cA_1\mr \cA_2$ whenever there exists a modal
refinement $R$ such that for all $s^0_1\in S^0_1$, there exists
$s^0_2\in S^0_2$ for which $( s^0_1, s^0_2)\in R$.  Sets of
implementations and thorough refinement are defined as for DMTS.  Note
that as both AA and DMTS implementations are LTS, it makes sense to use
thorough refinement and equivalence \emph{across} formalisms, writing
\eg~$\cA\treq \cD$ for an AA $\cA$ and a DMTS $\cD$.

\subsection{Hybrid modal logic}

We introduce a hybrid modal logic which can serve as compact
representation of AA.  This logic is closely related to the Boolean
modal transition systems of~\cite{DBLP:conf/atva/BenesKLMS11} and hybrid
in the sense of~\cite{book/Prior68,DBLP:journals/igpl/Blackburn00}: it
contains nominals, and the semantics of a nominal is given as all sets
which contain the nominal.

For a finite set $X$ of nominals, let $\cL( X)$ be the set of formulae
generated by the abstract syntax $\cL( X)\ni \phi\coloneqq \ttt\mid
\fff\mid \langle a\rangle x\mid \neg \phi\mid \phi\land \phi$, for $a\in
\Sigma$ and $x\in X$.  The semantics of a formula is a set of subsets of
$\Sigma\times X$, given as follows: $\lsem{ \ttt}= 2^{ \Sigma\times X}$,
$\lsem{ \fff}= \emptyset$, $\lsem{ \neg \phi}= 2^{ \Sigma\times
  X}\setminus \lsem{ \phi}$, $\lsem{ \langle a\rangle x}=\{ M\subseteq
\Sigma\times X\mid( a, x)\in M\}$, and $\lsem{ \phi\land \psi}= \lsem{
  \phi}\cap \lsem{ \psi}$.  We also define disjunction $\phi_1\lor
\phi_2= \neg( \phi_1\land \phi_2)$.

An \emph{$\cL$-expression} is a structure $\cE=( X, X^0, \Phi)$
consisting of finite sets $X^0\subseteq X$ of variables and a mapping
$\Phi: X\to \cL( X)$.  Such an expression is an \emph{implementation} if
$\lsem{ \Phi( x)}=\{ M\}$ is a singleton for each $x\in X$.  It can
easily be shown that $\cL$-implementations precisely correspond to LTS.

Let $\cE_1=( X_1, X^0_1, \Phi_1)$ and $\cE_2=( X_2, X^0_2, \Phi_2)$ be
$\cL$-expressions.  A relation $R\subseteq X_1\times X_2$ is a
\emph{modal refinement} if it holds for all $( x_1, x_2)\in R$ and all
$M_1\in \lsem{ \Phi_1( x_1)}$ that there exists $M_2\in \lsem{
  \Phi_2( x_2)}$ such that
\begin{compactitem}
\item $\forall( a, t_1)\in M_1: \exists( a, t_2)\in M_2:( t_1, t_2)\in
  R$,
\item $\forall( a, t_2)\in M_2: \exists( a, t_1)\in M_1:( t_1, t_2)\in
  R$.
\end{compactitem}
Again, we write $\cE_1\mr \cE_2$ whenever there exists 
such a modal refinement $R$ such that for all $x^0_1\in X^0_1$,
there exists $x^0_2\in X^0_2$ for which $( x^0_1, x^0_2)\in R$.
Sets of implementations and thorough refinement are defined as for DMTS.

\subsection{Structural equivalence}

We proceed to show that the three formalisms introduced in this section
are structurally equivalent.  Using the translations between AA and DMTS
discovered in~\cite{DBLP:conf/concur/BenesDFKL13} and new translations
between AA and hybrid logic, we show that these respect modal
refinement.

The translations $\bl$, $\lb$ between AA and our hybrid logic are
straightforward: For an AA $\cA=( S, S^0, \Tran)$ and all $s\in S$, let
\begin{equation*}
  \Phi( s)= \biglor_{ M\in \Tran( s)} \Big( \bigland_{( a, t)\in M}
  \langle a\rangle t\land \bigland_{( b, u)\notin M} \neg \langle
  b\rangle u\Big)
\end{equation*}
and define the $\cL$-expression $\bl( \cA)=( S, S^0, \Phi)$.

For an $\cL$-expression $\cE=( X, X^0, \Phi)$ and all $x\in X$, let
$\Tran( x)= \lsem{ \Phi( x)}$ and define the AA $\lb( \cE)=( X, X^0,
\Tran)$.

The translations $\db$, $\bd$ between DMTS and AA were discovered
in~\cite{DBLP:conf/concur/BenesDFKL13}.  For a DMTS $\cD=( S, S^0,
\omay, \omust)$ and all $s\in S$, let
\begin{equation*}
  \Tran(s)=\{ M\subseteq \Sigma\times S\mid \forall (a,t)\in M: s\may{a}
  t, \forall s\must{} N: N\cap M\ne \emptyset\}
\end{equation*}
and define the AA $\db( \cD)=( S, S^0, \Tran)$.\footnote{Note that there
  is an error in the corresponding formula
  in~\cite{DBLP:conf/concur/BenesDFKL13}.}

For an AA $\cA=( S, S^0, \Tran)$, define the DMTS $\bd( \cA)=( D, D^0,
\omay, \omust)$ as follows:
\begin{align*}
  D &= \{ M\in \Tran( s)\mid s\in S\} \\
  D^0 &= \{ M^0\in \Tran( s^0)\mid s^0\in S^0\} \\
  \omust &= \big\{\big( M,\{( a, M')\mid M'\in \Tran(
  t)\}\big)\bigmid( a, t)\in M\big\} \\
  \omay &= \{( M, a, M')\mid \exists M\must{} N: ( a, M')\in N\}
\end{align*}
Note that the state spaces of $\cA$ and $\bd( \cA)$ are not the same;
the one of $\bd( \cA)$ may be exponentially larger.  The following lemma
shows that this explosion is unavoidable:

\begin{lemma}
  \label{le:bfstodmtsblowup}
  There exists a one-state AA $\cA$ for which any DMTS $\cD\treq \cA$
  has at least $2^{ n- 1}$ states, where $n$ is the size of the alphabet
  $\Sigma$.
\end{lemma}

We notice that LTS are preserved by all translations: for any LTS $\cI$,
$\bl( \cI)= \lb( \cI)= \db( \cI)= \bd( \cI)= \cI$.
In~\cite{DBLP:conf/concur/BenesDFKL13} it is shown that the translations
between AA and DMTS respect sets of implementations, \ie~that $\db(
\cD)\treq \cD$ and $\bd( \cA)\treq \cA$ for all DMTS $\cD$ and all AA
$\cA$.  The next theorem shows that these and the other presented
translations respect modal refinement, hence these formalisms are not
only semantically equivalent, but \emph{structurally equivalent}.

\begin{theorem}
  \label{th:trans-modref}
  For all AA $\cA_1$, $\cA_2$, DMTS $\cD_1, \cD_2$ and $\cL$-expressions
  $\cE_1$, $\cE_2$:
  \begin{enumerate}
  \item $\cA_1\mr \cA_2$ iff $\bl( \cA_1)\mr \bl( \cA_2)$,
  \item $\cE_1\mr \cE_2$ iff $\lb( \cE_1)\mr \lb( \cE_2)$,
  \item $\cD_1\mr \cD_2$ iff $\db( \cD_1)\mr \db( \cD_2)$, and
  \item $\cA_1\mr \cA_2$ iff $\bd( \cA_1)\mr \bd( \cA_2)$.
  \end{enumerate}
\end{theorem}

\begin{proof}[sketch]
  We give a few hints about the proofs of the equivalences; the details
  can be found in appendix.  The first two equivalences follow easily
  from the definitions, once one notices that for both translations,
  $\lsem{ \Phi( x)}= \Tran( x)$ for all $x\in X$.  For the third
  equivalence, we can show that a DMTS modal refinement $\cD_1\mr \cD_2$
  is also an AA modal refinement $\db( \cD_1)\mr \db( \cD_2)$ and vice
  versa.

  The fourth equivalence is slightly more tricky, as the state space
  changes.  If $R\subseteq S_1\times S_2$ is an AA modal refinement
  relation witnessing $\cA_1\mr \cA_2$, then we can construct a DMTS
  modal refinement $R'\subseteq D_1\times D_2$, which witnesses $\bd(
  \cA_1)\mr \bd( \cA_2)$, by
  \begin{multline*}
    R'= \{( M_1, M_2)\mid \exists( s_1, s_2)\in R: M_1\in \Tran_1( s_1),
    M_2\in \Tran( s_2), \\
    \begin{aligned}
      & \forall( a, t_1)\in M_1: \exists( a, t_2)\in M_2:( t_1, t_2)\in
      R, \\
      & \forall( a, t_2)\in M_2: \exists( a, t_1)\in M_1:( t_1, t_2)\in
      R \}.
    \end{aligned}
  \end{multline*}
  Conversely, if $R\subseteq D_1\times D_2$ is a DMTS modal refinement
  witnessing $\bd( \cA_1)\mr \bd( \cA_2)$, then
  $R'\subseteq S_1\times S_2$ given by
  \begin{equation*}
    R'=\{( s_1, s_2)\mid \forall M_1\in \Tran_1( s_1): \exists M_2\in
    \Tran_2( s_2):( M_1, M_2)\in R\}
  \end{equation*}
  is an AA modal refinement. \qed
\end{proof}

The result on thorough equivalence
from~\cite{DBLP:conf/concur/BenesDFKL13} now easily follows:

\begin{corollary}
  \label{co:thorough1}
  For all AA $\cA$, DMTS $\cD$ and $\cL$-expressions $\cE$, $\bl(
  \cA)\treq \cA$, $\lb( \cE)\treq \cE$, $\db( \cD)\treq \cD$, and $\bd(
  \cA)\treq \cA$. \noproof
\end{corollary}

Also soundness of modal refinement for AA and hybrid logic follows
directly from Theorem~\ref{th:trans-modref}:

\begin{corollary}
  For all AA $\cA_1$ and $\cA_2$, $\cA_1\mr \cA_2$ implies $\cA_1\tr
  \cA_2$.  For all $\cL$-expressions $\cE_1$ and $\cE_2$, $\cE_1\mr
  \cE_2$ implies $\cE_1\tr \cE_2$. \noproof
\end{corollary}

\section{The Modal $\nu$-Calculus}
\label{se:mu}

We wish to extend the structural equivalences of the previous section to
the modal $\nu$-calculus.  Using translations between AA, DMTS and
$\nu$-calculus based on work
in~\cite{DBLP:conf/avmfss/Larsen89,DBLP:journals/tcs/BoudolL92}, it has
been shown in~\cite{DBLP:conf/concur/BenesDFKL13} that $\nu$-calculus
and DMTS/AA are \emph{semantically} equivalent.  We will see below that
there is a \emph{mismatch} between the translations
from~\cite{DBLP:conf/concur/BenesDFKL13} (and hence between the
translations
in~\cite{DBLP:conf/avmfss/Larsen89,DBLP:journals/tcs/BoudolL92}) which
precludes structural equivalence and then proceed to propose a new
translation which fixes the mismatch.

\subsection{Syntax and semantics}

We first recall the syntax and semantics of the modal $\nu$-calculus,
the fragment of the modal
$\mu$-calculus~\cite{unpub/ScottB69,DBLP:journals/tcs/Kozen83} with only
maximal fixed points.  Instead of an explicit maximal fixed point
operator, we use the representation by equation systems in
Hennessy-Milner logic developed in~\cite{DBLP:journals/tcs/Larsen90}.

For a finite set $X$ of variables, let $\HML( X)$ be the set of
\emph{Hennessy-Milner formulae}, generated by the abstract syntax $\HML(
X)\ni \phi\Coloneqq \ttt\mid \fff\mid x\mid \langle a\rangle \phi\mid[
a] \phi\mid \phi\land \phi\mid \phi\lor \phi$, for $a\in \Sigma$ and
$x\in X$.

A \emph{declaration} is a mapping $\Delta: X\to \HML( X)$; we recall the
maximal fixed point semantics of declarations
from~\cite{DBLP:journals/tcs/Larsen90}.  Let $( S, S^0, \omust)$ be an
LTS, then an \emph{assignment} is a mapping $\sigma: X\to 2^S$.  The set
of assignments forms a complete lattice with order $\sigma_1\sqsubseteq
\sigma_2$ iff $\sigma_1( x)\subseteq \sigma_2( x)$ for all $x\in X$ and
lowest upper bound $\big(\bigsqcup_{ i\in I} \sigma_i\big)( x)=
\bigcup_{ i\in I} \sigma_i( x)$.

The semantics of a formula is a subset of $S$, given relative to an
assignment~$\sigma$, defined as follows: $\lsem \ttt \sigma= S$, $\lsem
\fff \sigma= \emptyset$, $\lsem x \sigma= \sigma( x)$, $\lsem{ \phi\land
  \psi} \sigma= \lsem \phi\sigma\cap \lsem \psi \sigma$, $\lsem{
  \phi\lor \psi} \sigma= \lsem \phi\sigma\cup \lsem \psi \sigma$, and
\begin{align*}
  \lsem{\langle a\rangle \phi} \sigma &= \{ s\in S\mid \exists s\must a
  s': s'\in \lsem \phi \sigma\}, \\
  \lsem{[ a] \phi} \sigma &= \{ s\in S\mid \forall s\must a s': s'\in
  \lsem \phi \sigma\}.
\end{align*}
The semantics of a declaration $\Delta$ is then the assignment defined
by
\begin{equation*}
  \lsem \Delta= \bigsqcup\{ \sigma: X\to 2^S\mid \forall x\in X:
  \sigma( x)\subseteq \lsem{ \Delta( x)} \sigma\};
\end{equation*}
the maximal (pre)fixed point of $\Delta$.

A \emph{$\nu$-calculus expression} is a structure $\cN=( X, X^0,
\Delta)$, with $X^0\subseteq X$ sets of variables and $\Delta: X\to
\HML( X)$ a declaration.  We say that an LTS $\cI=( S, S^0, \omust)$
\emph{implements} (or models) the expression, and write $\cI\models
\cN$, if it holds that for all $s^0\in S^0$, there is $x^0\in X^0$ such
that $s^0\in \lsem \Delta( x^0)$.  We write $\sem \cN$ for the set of
implementations (models) of a $\nu$-calculus expression $\cN$.  As for
DMTS, we write $\sem x= \sem{( X,\{ x\}, \Delta)}$ for $x\in X$, and
thorough refinement of expressions and states is defined accordingly.

The following lemma introduces a \emph{normal form} for $\nu$-calculus
expressions:

\begin{lemma}
  \label{le:hmlnormal}
  For any $\nu$-calculus expression $\cN_1=( X_1, X^0_1, \Delta_1)$,
  there exists another expression $\cN_2=( X_2, X^0_2, \Delta_2)$ with
  $\sem{ \cN_1}= \sem{ \cN_2}$ and such that for any $x\in X$,
  $\Delta_2( x)$ is of the form
  \begin{equation}
    \label{eq:hmlnormal}
    \Delta_2( x)= \bigland_{ i\in I}\Big( \biglor_{ j\in
      J_i} \langle a_{ ij}\rangle  x_{ ij}\Big)\land \bigland_{ a\in
      \Sigma}[ a] \Big( \biglor_{ j\in J_a} y_{ a, j}\Big)
  \end{equation}
  for finite (possibly empty) index sets $I$, $J_i$, $J_a$, for $i\in I$
  and $a\in \Sigma$, and all $x_{ ij}, y_{ a, j}\in X_2$.  Additionally,
  for all $i\in I$ and $j\in J_i$, there exists $j'\in J_{ a_{ ij}}$ for
  which $x_{ ij}\tr y_{ a_{ ij}, j'}$.
\end{lemma}

As this is a type of \emph{conjunctive normal form}, it is clear that
translating a $\nu$-calculus expression into normal form may incur an
exponential blow-up.

We introduce some notation for $\nu$-calculus expressions in normal form
which will make our life easier later.  Let $\cN=( X, X^0, \Delta)$ be
such an expression and $x\in X$, with $\Delta( x)= \bigland_{ i\in
  I}\big( \biglor_{ j\in J_i} \langle a_{ ij}\rangle x_{ ij}\big)\land
\bigland_{ a\in \Sigma}[ a] \big( \biglor_{ j\in J_a} y_{ a, j}\big)$ as
in the lemma.  Define $\Diamond( x)=\{\{( a_{ ij}, x_{ ij})\mid j\in
J_i\}\mid i\in I\}$ and, for each $a\in \Sigma$, $\Box^a( x)=\{ y_{ a,
  j}\mid j\in J_a\}$.  Note that now $\Delta( x)= \bigland_{ N\in
  \Diamond(x)} \big( \biglor_{( a, y)\in N} \langle a\rangle y\big)
\land \bigland_{ a\in \Sigma}[ a]\big( \biglor_{ y\in \Box^a( x)}
y\big)$.

\subsection{Refinement}

In order to expose our structural equivalence, we need to introduce a
notion of modal refinement for the modal $\nu$-calculus.  For reasons
which will become apparent later, we define two different such notions:

Let $\cN_1=( X_1, X^0_1, \Delta_1)$, $\cN_2=( X_2, X^0_2, \Delta_2)$ be
$\nu$-calculus expressions in normal form and $R\subseteq X_1\times
X_2$.  The relation $R$ is a \emph{modal refinement} if it holds for all
$( x_1, x_2)\in R$ that
\begin{enumerate}
\item \label{en:nucalref.may} for all $a\in \Sigma$ and every $y_1\in
  \Box^a_1( x_1)$, there is $y_2\in \Box^a_2( x_2)$ for which $( y_1,
  y_2)\in R$, and
\item for all $N_2\in \Diamond_2( x_2)$ there is $N_1\in \Diamond_1(
  x_1)$ such that for each $( a, y_1)\in N_1$, there exists $( a,
  y_2)\in N_2$ with $( y_1, y_2)\in R$.
\end{enumerate}
$R$ is a \emph{modal-thorough} refinement if, instead
of~\ref{en:nucalref.may}., it holds that
\begin{enumerate}
\item[1$'$.] for all $a\in \Sigma$, all $y_1\in \Box^a_1( x_1)$ and
  every $y_1'\in X_1$ with $y_1'\tr y_1$, there is $y_2\in \Box^a_2(
  x_2)$ and $y_2'\in X_2$ such that $y_2'\tr y_2$ and $( y_1', y_2')\in
  R$.
\end{enumerate}
We say that $\cN_1$ \emph{refines} $\cN_2$ whenever there exists such a
refinement $R$ such that for every $x^0_1\in X^0_1$ there exists
$x^0_2\in X^0_2$ for which $( x^0_1, x^0_2)\in R$.  We write $\cN_1\mr
\cN_2$ in case of modal and $\cN_1\mtr \cN_2$ in case of modal-thorough
refinement.

We remark that whereas modal refinement for $\nu$-calculus expressions
is a simple and entirely syntactic notion, modal-thorough refinement
involves semantic inclusions of states.  Using results
in~\cite{DBLP:journals/iandc/BenesKLS12}, this implies that modal
refinement can be decided in time polynomial in the size of the
(normal-form) expressions, whereas deciding modal-thorough refinement is
EXPTIME-complete.

\subsection{Translation from DMTS to $\nu$-calculus}

Our translation from DMTS to $\nu$-calculus is new, but similar to the
translation from AA to $\nu$-calculus given
in~\cite{DBLP:conf/concur/BenesDFKL13}.  This in turn is based on the
\emph{characteristic formulae} of~\cite{DBLP:conf/avmfss/Larsen89} (see
also~\cite{books/AcetoILS07}).

For a DMTS $\cD=( S, S^0, \omay, \omust)$ and all $s\in S$, we define
$\Diamond(s)=\{ N  \mid s \must{} N\}$ and, for each $a\in \Sigma$,
$\Box^a(s)=\{ t \mid s \may{a} t\}$.  Then, let
\begin{equation*}
  \Delta( s)= \bigland_{N \in \Diamond(s)}\Big( \biglor_{( a, t)\in N}
  \langle a\rangle t\Big)\land \bigland_{ a\in \Sigma}[ a]\Big(
  \biglor_{t \in \Box^a(s)} t\Big)
\end{equation*}
and define the (normal-form) $\nu$-calculus expression $\ddh( \cD)=( S,
S^0, \Delta)$.

Note how the formula precisely expresses that we demand at least one of
every choice of disjunctive must-transitions (first part) and permit all
may-transitions (second part); this is also the intuition of the
characteristic formulae of~\cite{DBLP:conf/avmfss/Larsen89}.  Using
results of~\cite{DBLP:conf/concur/BenesDFKL13} (which introduces a very
similar translation from AA to $\nu$-calculus expressions), we see that
$\ddh( \cD)\treq \cD$ for all DMTS $\cD$.

\begin{theorem}
  \label{th:dmts->nu}
  For all DMTS $\cD_1$ and $\cD_2$, $\cD_1\mr \cD_2$ iff $\ddh(
  \cD_1)\mr \ddh( \cD_2)$.
\end{theorem}

\begin{proof}
  For the forward direction, let $R \subseteq S_1 \times S_2$ be a modal
  refinement between $\cD_1=(S_1,S_1^0, \omay_1, \omust_1)$ and
  $\cD_2=(S_2,S_2^0, \omay_2, \omust_2)$; we show that $R$ is also a
  modal refinement between $\ddh(\cD_1)=(S_1,S_1^0, \Delta_1)$ and
  $\ddh(\cD_2)=(S_2,S_2^0, \Delta_2)$.  Let $(s_1,s_2) \in R$.
  \begin{compactitem}
  \item Let $a \in \Sigma$ and $t_1 \in \Box_1^a(s_1)$, then $s_1
    \may{a}_1 t_1$, which implies that there is $t_2 \in S_2$ for which
    $s_2 \may{a}_2 t_2$ and $(t_1,t_2) \in R$.  By definition of
    $\Box^a_2$, $t_2 \in \Box_2^a(s_2)$.
  \item Let $N_2 \in \Diamond_2(s_2)$, then $s_2 \must{}_2 N_2$, which
    implies that there exists $s_1 \must{}_1 N_1$ such that $\forall
    (a,t_1) \in N_1: \exists (a,t_2) \in N_2:( t_1,t_2) \in R$.  By
    definition of $\Box_1^a$, $N_1 \in \Box_1^a(s_1)$.
  \end{compactitem}

  For the other direction, let $R \subseteq S_1 \times S_2$ be a modal
  refinement between $\ddh(\cD_1)$ and $\ddh(\cD_2)$, we show that $R$
  is also a modal refinement between $\cD_1$ and $\cD_2$.  Let
  $(s_1,s_2) \in R$.
  \begin{compactitem}
  \item For all $s_1 \may{a}_1 t_1$, $t_1 \in \Box_1^a(s_1)$, which
    implies that there is $t_2 \in \Box_2^a(s_2)$ with $(t_1,t_2) \in
    R$, and by definition of $\Box_2^a$, $s_2 \may{a}_2 t_2$.
 
  \item For all $s_2 \must{}_2 N_2$, $N_2 \in \Diamond_2(s_2)$, which
    implies that there is $N_1 \in \Diamond_1(s_1)$ such that $\forall
    (a,t_1) \in N_1: \exists (a,t_2) \in N_2:( t_1,t_2) \in R$, and by
    definition of $\Box_1^a$, $s_1 \must{}_1 N_1$. \qed
  \end{compactitem}
\end{proof}

\subsection{Old translation from $\nu$-calculus to DMTS}

We recall the translation from $\nu$-calculus to DMTS given
in~\cite{DBLP:conf/concur/BenesDFKL13}, which is based on a translation
from Hennessy-Milner formulae (without recursion and fixed points) to
sets of acyclic MTS in~\cite{DBLP:journals/tcs/BoudolL92}.
For a $\nu$-calculus expression $\cN=( X, X^0, \Delta)$ in normal form, let
\begin{gather*}
  \omay=\{( x, a, y')\in X\times \Sigma\times X\mid \exists y\in
  \Box^a( x): y'\tr y\}, \\
  \omust=\{( x, N)\mid x\in X, N\in \Diamond( x)\}.
\end{gather*}
and define the DMTS $\hdt( \cN)=( X, X^0, \omay, \omust)$.

Note how this translates diamonds to disjunctive must-transitions
directly, but for boxes takes semantic inclusions into account: for a
subformula $[ a] y$, may-transitions are created to all variables which
are semantically below $y$.  This is consistent with the interpretation
of formulae-as-properties: $[ a] y$ means ``for any $a$-transition,
$\Delta( y)$ must hold''; but $\Delta( y)$ holds for all variables which
are semantically below $y$.

It follows from results in~\cite{DBLP:conf/concur/BenesDFKL13} (which
uses a slightly different normal form for $\nu$-calculus expressions)
that $\hdt( \cN)\treq \cN$ for all $\nu$-calculus expressions $\cN$.

\begin{theorem}
  \label{th:nu->dmts-t}
  For all $\nu$-calculus expressions, $\cN_1\mtr \cN_2$ iff $\hdt(
  \cN_1)\mr \hdt( \cN_2)$.
\end{theorem}

\begin{proof}
  For the forward direction, let $R \subseteq X_1 \times X_2$ be a
  modal-thorough refinement between $\cN_1=( X_1, X_1^0, \Delta_1)$ and
  $\cN_2=( X_2, X_2^0, \Delta_2)$.  We show that $R$ is also a modal
  refinement between $\hdt(\cN_1)=( X_1, X_1^0, \omay_1, \omust_2)$ and
  $\hdt(\cN_2)=( X_2, X_2^0, \omay_2,\omust_2)$.  Let $( x_1, x_2)\in
  R$.
  \begin{compactitem}
  \item Let $x_1 \may{a}_1 y'_1$. By definition of $\omay_1$, there is
    $y_1 \in \Box_1^a(x_1)$ for which $y'_1 \tr y_1$.  Then by
    modal-thorough refinement, this implies that there exists $y_2 \in
    \Box_2^a(x_2)$ and $y'_2 \in X_2$ such that $y'_2 \tr y_2$ and
    $(y'_1,y'_2) \in R$.  By definition of $\omay_2$ we have $x_2
    \may{a}_2 y'_2$.
  \item Let $x_2 \must{}_2 N_2$, then we have $N_2 \in \Diamond_2(x_2)$.
    By modal-thorough refinement, this implies that there is $N_1 \in
    \Diamond_1(x_1)$ such that $\forall (a,y_1) \in N_1: \exists (a,y_2)
    \in N_2:( y_1,y_2) \in R$.  By definition of $\omust_1$, $x_1
    \must{}_1 N_1$.
  \end{compactitem}

  Now to the proof that $\hdt( \cN_1)\mr \hdt( \cN_2)$ implies
  $\cN_1\mtr \cN_2$.  We have a modal refinement (in the DMTS sense)
  $R\subseteq X_1\times X_2$.  We must show that $R$ is also a
  modal-thorough refinement.  Let $( x_1, x_2)\in R$.
  \begin{compactitem}
  \item Let $a\in \Sigma$, $y_1 \in \Box_1^a(x_1)$ and $y'_1 \in X_1$
    such that $y'_1 \tr y_1$. Then by definition of $\omay_1$, $x_1
    \may{a}_1 y'_1$.  By modal refinement, this implies that there
    exists $x_2 \may{a}_2 y'2$ with $(y'_1,y'_2) \in R$. Finally, by
    definition of $\omay_2$, there exists $y_2 \in \Box_2^a(x_2)$ such
    that $y'_2 \tr y_2$.
  \item Let $N_2 \in \Diamond_2(x_2)$, then by definition of $\omust_2$,
    $x_2 \must{}_2 N_2$.  Then, by modal refinement, this implies that
    there exists $x_1 \must{}_1 N_1$ such that $\forall (a,y_1) \in N_1:
    \exists (a,y_2) \in N_2:( y_1,y_2) \in R$.  By definition of
    $\omust_1$, $N_1 \in \Box_1^a(x_1)$. \qed
  \end{compactitem}
\end{proof}

\subsection{Discussion}

Notice how Theorems~\ref{th:dmts->nu} and~\ref{th:nu->dmts-t} expose a
\emph{mismatch} between the translations: $\ddh$ relates DMTS refinement
to $\nu$-calculus \emph{modal} refinement, whereas $\hdt$ relates it to
\emph{modal-thorough} refinement.  Both translations are well-grounded
in the literature and well-understood,
\cf~\cite{DBLP:conf/concur/BenesDFKL13,DBLP:journals/tcs/BoudolL92,DBLP:conf/avmfss/Larsen89},
but this mismatch has not been discovered up to now.  Given that the
above theorems can be understood as universal properties of the
translations, it means that there is no notion of refinement for
$\nu$-calculus which is consistent with them both.

The following lemma, easily shown by inspection, shows that this
discrepancy is related to the may-completion for DMTS:

\begin{lemma}
  \label{le:mc=hdt*dh}
  For any DMTS $\cD$, $\maycomp{\cD} = \hdt( \dh (\cD))$. \noproof
\end{lemma}

As a corollary, we see that modal refinement and modal-thorough
refinement for $\nu$-calculus are incomparable: Referring back to
Example~\ref{ex:mc_prop}, we have $\cD \mr \cD_1$, hence by
Theorem~\ref{th:dmts->nu}, $\ddh(\cD) \mr \ddh(\cD_1)$.  On the other
hand, we know that $\maycomp{\cD} \not\mr \maycomp{\cD_1}$, \ie~by
Lemma~\ref{le:mc=hdt*dh}, $\hdt(\ddh(\cD)) \not\mr \hdt(\ddh(\cD_1))$,
and then by Theorem~\ref{th:nu->dmts-t}, $\ddh(\cD) \not\mtr
\ddh(\cD_1)$.

To expose an example where modal-thorough refinement holds, but modal
refinement does not, we note that $\maycomp{\cD_2} \mr \maycomp{\cD}$
implies, again using Lemma~\ref{le:mc=hdt*dh} and
Theorem~\ref{th:nu->dmts-t}, that $\ddh(\cD_2) \mtr \ddh(\cD)$.  On the
other hand, we know that $\cD_2 \not\mr \cD$, so by
Theorem~\ref{th:dmts->nu}, $\ddh(\cD_2) \not\mr \ddh(\cD)$.

\subsection{New translation from $\nu$-calculus to DMTS}

We now show that the mismatch between DMTS and $\nu$-calculus
expressions can be fixed by introducing a new, simpler translation from
$\nu$-calculus to DMTS.

For a $\nu$-calculus expression $\cN=( X, X^0, \Delta)$ in normal form,
let
\begin{gather*}
  \omay=\{( x, a, y)\in X\times \Sigma\times X\mid y\in \Box^a( x)\}, \\
  \omust=\{( x, N)\mid x\in X, N\in \Diamond( x)\}.
\end{gather*}
and define the DMTS $\hd( \cN)=( X, X^0, \omay, \omust)$.  This is a
simple syntactic translation: boxes are translated to disjunctive
must-transitions and diamonds to may-transitions.

\begin{theorem}
  \label{th:nu->dmts}
  For all $\nu$-calculus expressions, $\cN_1\mr \cN_2$ iff $\hd(
  \cN_1)\mr \hd( \cN_2)$.
\end{theorem}

\begin{proof}
  Let $R \subseteq X_1 \times X_2$ be a modal refinement between
  $\cN_1=(X_1,X_1^0,\Delta_1)$ and $\cN_2=(X_2,X_2^0,\Delta_2)$; we show
  that $R$ is also a modal refinement between $\hd(\cN_1)=(S_1,S_1^0,
  \omay_1, \omust_1)$ and $\hd(\cN_2)=(S_2,S_2^0, \omay_2, \omust_2)$.
  Let $(x_1,x_2) \in R$.
  \begin{compactitem}
  \item Let $x_1 \may{a}_1 y_1$, then $y_1 \in \Box_1^a(x_1)$, which
    implies that there exists $y_2 \in \Box_2^a(x_2)$ for which
    $(y_1,y_2) \in R$, and by definition of $\omay_2$, $x_2 \may{a}_2
    y_2$.
  \item Let $x_2 \must{}_2 N_2$, then $N_2 \in \Diamond_2(x_2)$, hence
    there is $N_1 \in \Diamond_1(x_1)$ such that $\forall (a,y_1) \in
    N_1: \exists (a,y_2) \in N_2:( y_1,y_2) \in R$, and by definition of
    $\omust_1$, $x_1 \must{}_1 N_1$.
  \end{compactitem}

  Now let $R \subseteq X_1 \times X_2$ be a modal refinement between
  $\hd(\cN_1)$ and $\hd(\cN_2)$, we show that $R$ is also a modal
  refinement between $\cN_1$ and $\cN_2$.  Let $(x_1,x_2) \in R$,
  \begin{compactitem}
  \item Let $a \in \Sigma$ and $y_1 \in \Box_1^a(x_1)$.  Then $x_1
    \may{a}_1 y_1$, which implies that there is $y_2 \in X_2$ for which
    $x_2 \may{a}_2 y_2$ and $(y_1,y_2) \in R$, and by definition of
    $\omay_2$, $t_2 \in \Box_2^a(s_2)$.
  \item Let $N_2 \in \Diamond_2(x_2)$, then $x_2 \must{}_2 N_2$, so
    there is $x_1 \must{}_1 N_1$ such that $\forall (a,y_1) \in N_1:
    \exists (a,y_2) \in N_2:( y_1,y_2) \in R$.  By definition of
    $\omust_1$, $N_1 \in \Box_1^a(x_1)$. \qed
  \end{compactitem}
\end{proof}

We finish the section by proving that also for the syntactic translation
$\hd( \cN)\treq \cN$ for all $\nu$-calculus expressions; this shows that
our translation can serve as a replacement for the partly-semantic
$\hdt$ translation
from~\cite{DBLP:conf/concur/BenesDFKL13,DBLP:journals/tcs/BoudolL92}.
First we remark that $\ddh$ and $\hd$ are inverses to each other:

\begin{proposition}
  \label{pr:id=dh*hd}
  For any $\nu$-calculus expression $\cN$, $\ddh(\hd(\cN))=\cN$; for any
  DMTS $\cD$, $\hd(\ddh(\cD))=\cD$. \noproof
\end{proposition}

\begin{corollary}
  \label{co:hd-treq}
  For all $\nu$-calculus expressions $\cN$, $\hd( \cN)\treq \cN$. \noproof
\end{corollary}

\section{The Modal $\nu$-Calculus as a Specification Theory}
\label{se:spectheory}

Now that we have exposed a close structural correspondence between the
modal $\nu$-calculus and DMTS, we can transfer the operations which make
DMTS a complete specification theory to the $\nu$-calculus.

\subsection{Refinement and implementations}

As for DMTS and AA, we can define an embedding of LTS into the modal
$\nu$-calculus so that implementation $\models$ and refinement $\mr$
coincide.  We say that a $\nu$-calculus expression $( X, X^0, \Delta)$
in normal form is an \emph{implementation} if $\Diamond( x)=\{\{( a,
y)\}\mid y\in \Box^a( x), a\in \Sigma\}$ for all $x\in X$.

The $\nu$-calculus translation of a LTS $( S, S^0, \omust)$ is the
expression $( S, S^0, \Delta)$ in normal form with $\Diamond( s)=\{\{(
a, t)\}\mid s\must a t\}$ and $\Box^a( s)=\{ t\mid s\must a t\}$.  This
defines a bijection between LTS and $\nu$-calculus implementations.

\begin{theorem}
  For any LTS $\cI$ and any $\nu$-calculus expression $\cN$, $\cI\models
  \cN$ iff $\cI\mr \cN$.
\end{theorem}

\begin{proof}
  $\cI\models \cN$ is the same as $\cI\in \sem{ \cN}$, which by
  Corollary~\ref{co:hd-treq} is equivalent to $\cI\in \sem{ \hd( \cN)}$.
  By definition, this is the same as $\cI\mr \hd( \cN)$, which using
  Theorem~\ref{th:nu->dmts} is equivalent to $\cI\mr \cN$. \qed
\end{proof}

Using transitivity, this implies that modal refinement for
$\nu$-calculus is sound:

\begin{corollary}
  For all $\nu$-calculus expressions, $\cN_1\mr \cN_2$ implies $\cN_1\tr
  \cN_2$.\hfill$\Box$
\end{corollary}

\subsection{Disjunction and conjunction}

As for DMTS, disjunction of $\nu$-calculus expressions is
straight-forward.  Given $\nu$-calculus expressions $\cN_1=( X_1, X^0_1,
\Delta_1)$, $\cN_2=( X_2, X^0_2, \Delta_2)$ in normal form, their
\emph{disjunction} is $\cN_1\lor \cN_2=( X_1\cup X_2, X_1^0\cup X_2^0,
\Delta)$ with $\Delta( x_1)= \Delta_1( x_1)$ for $x_1\in X_1$ and
$\Delta( x_2)= \Delta_2( x_2)$ for $x_2\in X_2$.

The \emph{conjunction} of $\nu$-calculus expressions like above is
$\cN_1\land \cN_2=( X, X^0, \Delta)$ defined by $X= X_1\times X_2$,
$X^0= X_1^0\times X_2^0$, $\Box^a( x_1, x_2)= \Box_1^a( x_1)\times
\Box_2^a( x_2)$ for each $( x_1, x_2)\in X$, $a\in \Sigma$, and for each
$( x_1, x_2)\in X$,
\begin{align*}
  \Diamond( x_1, x_2) &= \big\{ \{( a,( y_1, y_2))\mid( a, y_1)\in N_1,
  ( y_1, y_2)\in \Box^a( x_1, x_2)\}\bigmid N_1\in \Diamond_1(
  x_1)\big\} \\
  &\;\cup \big\{ \{( a,( y_1, y_2))\mid( a, y_2)\in N_2, ( y_1, y_2)\in
  \Box^a( x_1, x_2)\}\bigmid N_2\in \Diamond_2( x_2)\big\}.
\end{align*}

Note that both $\cN_1\lor \cN_2$ and $\cN_1\land \cN_2$ are again
$\nu$-calculus expressions in normal form.

\begin{theorem}
  \label{th:logops}
  For all $\nu$-calculus expressions $\cN_1$, $\cN_2$, $\cN_3$ in normal
  form,
  \begin{compactitem}
  \item $\cN_1\lor \cN_2\mr \cN_3$ iff $\cN_1\mr \cN_3$ and $\cN_2\mr
    \cN_3$,
  \item $\cN_1\mr \cN_2\land \cN_3$ iff $\cN_1\mr \cN_2$ and $\cN_1\mr
    \cN_3$,
  \item $\sem{ \cN_1\lor \cN_2}= \sem{ \cN_1}\cup \sem{ \cN_2}$, and
    $\sem{ \cN_1\land \cN_2}= \sem{ \cN_1}\cap \sem{ \cN_2}$.
  \end{compactitem}
\end{theorem}

\begin{theorem}
  \label{th:distlat}
  With operations $\lor$ and $\land$, the class of $\nu$-calculus
  expressions forms a bounded distributive lattice up to
  $\mreq$.
\end{theorem}

The bottom element (up to $\mreq$) in the lattice is the empty
$\nu$-calculus expression $\bot=( \emptyset, \emptyset, \emptyset)$, and
the top element (up to $\mreq$) is $\top=(\{ s\},\{ s\}, \Delta)$ with
$\Delta( s)= \ttt$.

\subsection{Structural composition}

The structural composition operator for a specification theory is to
mimic, at specification level, the structural composition of
implementations.  That is to say, if $\para{}$ is a composition operator for
implementations (LTS), then the goal is to extend $\para{}$ to specifications
such that for all specifications $\cS_1$, $\cS_2$,
\begin{equation}
  \label{eq:compcomplete}
  \sem{ \cS_1\para{} \cS_2}= \big\{ \cI_1\para{} \cI_2\mid \cI_1\in
  \sem{ \cS_1}, \cI_2\in \sem{ \cS_2}\big\}.
\end{equation}

For simplicity, we use CSP-style synchronization for structural
composition of LTS, however, our results readily carry over to other
types of composition.  Analogously to the situation for
MTS~\cite{DBLP:journals/tcs/BenesKLS09}, we have the following negative
result:

\begin{theorem}
  There is no operator $\para{}$ for the $\nu$-calculus which
  satisfies~\eqref{eq:compcomplete}.
\end{theorem}

\begin{proof}
  We first note that due to Theorem~\ref{th:logops}, it is the case that
  implementation sets of $\nu$-calculus expressions are closed under
  disjunction: for any $\nu$-calculus expression $\cN$ and $\cI_1,
  \cI_2\in \sem{ \cN}$, also $\cI_1\lor \cI_2\in \sem{ \cN}$.

  Now assume there were an operator as in the theorem, then because of
  the translations, \eqref{eq:compcomplete} would also hold for DMTS.
  Hence for all DMTS $\cD_1$, $\cD_2$, $\{ \cI_1\para{} \cI_2\mid \cI_1\in
  \sem{ \cD_1}, \cI_2\in \sem{ \cD_2}\}$ would be closed under
  disjunction.  But Example 7.8 in~\cite{DBLP:journals/tcs/BenesKLS09}
  exhibits two DMTS (actually, MTS) for which this is not the case, a
  contradiction. \qed
\end{proof}

Given that we cannot have~\eqref{eq:compcomplete}, the revised goal is
to have a \emph{sound} composition operator for which the right-to-left
inclusion holds in~\eqref{eq:compcomplete}.  We can obtain one such from
the structural composition of AA introduced
in~\cite{DBLP:conf/concur/BenesDFKL13}.  We hence define, for
$\nu$-calculus expressions $\cN_1=( X_1, X_1^0, \Delta_1)$, $\cN_2=(
X_2, X_2^0, \Delta_2)$ in normal form, $\cN_1\para{} \cN_2= \bh( \hb(
\cN_1)\para A \hb( \cN_2))$, where $\para A$ is AA
composition and we write $\bh= \ddh\circ \bd$ and $\hb= \db\circ \hd$
for the composed translations.

Notice that the involved translation from AA to DMTS may
lead to an exponential blow-up.  Unraveling the definition gives us the
following explicit expression for $\cN_1\para{} \cN_2=( X, X^0,
\Delta)$:
\begin{compactitem}
\item $X=\big\{\{( a,( y_1, y_2))\mid \forall i\in\{ 1, 2\}:( a, y_i)\in
  M_i\}\bigmid \forall i\in\{ 1, 2\}: M_i\subseteq \Sigma\times X_i,
  \exists x_i\in X_i: \forall( a, y_i')\in M_i: y_i'\in \Box_i^a( x_i),
  \forall N_i\in \Diamond_i( x_i): N_i\cap M_i\ne \emptyset\big\}$,
\item $X^0=\big\{\{( a,( y_1, y_2))\mid \forall i\in\{ 1, 2\}:( a,
  y_i)\in M_i\}\bigmid \forall i\in\{ 1, 2\}: M_i\subseteq \Sigma\times
  X_i, \exists x_i\in X_i^0: \forall( a, y_i')\in M_i: y_i'\in \Box_i^a(
  x_i), \forall N_i\in \Diamond_i( x_i): N_i\cap M_i\ne
  \emptyset\big\}$,
\item $\Diamond( x)=\big\{\{( a,\{( b,( z_1, z_2))\mid \forall i\in\{ 1,
  2\}:( b, z_i)\in M_i\}\mid \forall i\in\{ 1, 2\}: M_i\subseteq
  \Sigma\times X_i, \forall( a, z_i')\in M_i: z_i'\in \Box_i^b( y_i),
  \forall N_i\in \Diamond_i( y_i): N_i\cap M_i\ne \emptyset\}\bigmid(
  a,( y_1, y_2))\in x\big\}$ for each $x\in X$, and
\item $\Box^a( x)=\{ y\mid \exists N\in \Diamond( x):( a, y)\in N\}$.
\end{compactitem}

\begin{theorem}
  \label{th:indimp}
  For all $\nu$-calculus expressions $\cN_1$, $\cN_2$, $\cN_3$, $\cN_4$
  in normal form, $\cN_1\mr \cN_3$ and $\cN_2\mr \cN_4$ imply
  $\cN_1\para{} \cN_2\mr \cN_3\para{} \cN_4$.
\end{theorem}

\begin{proof}
  This follows directly from the analogous property for
  AA~\cite{DBLP:conf/concur/BenesDFKL13} and the translation
  theorems~\ref{th:trans-modref}, \ref{th:dmts->nu}
  and~\ref{th:nu->dmts}. \qed
\end{proof}

This implies the right-to-left inclusion in~\eqref{eq:compcomplete},
\ie~$\big\{ \cI_1\para{} \cI_2\mid \cI_1\in \sem{ \cN_1}, \cI_2\in \sem{
  \cN_2}\big\}\subseteq \sem{ \cN_1\para{} \cN_2}$.  It also entails
\emph{independent implementability}, in that the structural composition
of the two refined specifications $\cN_1$, $\cN_2$ is a refinement of
the composition of the original specifications $\cN_3$, $\cN_4$.
Fig.~\ref{fi:comp} shows an example of the DMTS analogue of this
structural composition.

\begin{figure}[tp]
  \centering
  \begin{tikzpicture}[->, >=stealth', font=\footnotesize,
    state/.style={shape=circle, draw, initial text=,inner
      sep=.5mm,minimum size=5mm}, scale=.8]
    \begin{scope}
      \node at (0,1.5) {$\cD_1$};
      \node[state,initial] (X) at (0,0) {$s_1$};
      \node[state] (Z) at (2,.7) {$t_1$};
      \node[state] (Y) at (2,-.7) {$u_1$};
      \coordinate (XX) at (.8,0);
      \path[-] (X) edge (XX); 
      \path (XX) edge[bend left] node[above]{$a$} (Z); 
      \path (XX) edge[bend right] node[below]{$b$} (Y); 
      \path (Z) edge[loop right] node[right]{$a$} (Z);
    \end{scope}
    \begin{scope}[xshift=15em]
      \node at (0,1.5) {$\cD_2$};
      \node[state,initial] (X) at (0,.7) {$s_2$};
      \node[state] (Z) at (2,.7) {$t_2$};
      \node[state] (Y) at (0,-.7) {$u_2$};
      \path (X) edge node[above]{$a$} (Z); 
      \path (X) edge node[right]{$a$} (Y); 
      \path (Z) edge[loop right] node[right]{$a$} (Z);
    \end{scope}
    \begin{scope}[xshift=30em]
      \node at (0,1.5) {$\cD_1\para{} \cD_2$};
      \node[state, initial] (X) at (0,.7) {$s'$};
      \node[state] (Z) at (2,.7) {$t'$};
      \node[state, initial] (Y) at (0,-.7) {$u'$};
      \path (X) edge node[above]{$a$} (Z); 
      \path (X) edge node[right]{$a$} (Y); 
      \path (Z) edge[loop right] node[right]{$a$} (Z);
    \end{scope}
  \end{tikzpicture}
  \caption{%
    \label{fi:comp}
    DMTS $\cD_1$, $\cD_2$ and the reachable parts of their structural
    composition $\cD_1\para{} \cD_2$.  Here, $s'=\{( a,( t_1, t_2)),(
    a,( t_1, u_2))\}$, $t'=\{( a,( t_1, t_2))\}$ and $u'= \emptyset$.
    Note that $\cD_1\para{} \cD_2$ has two initial states.}
\end{figure}
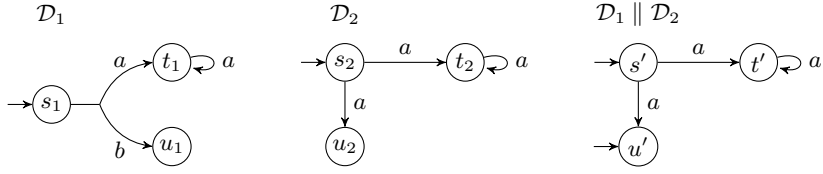

\subsection{Quotient}

The quotient operator $\by{}$ for a specification theory is used to
synthesize specifications for components of a structural composition.
Hence it is to have the property, for all specifications $\cS$, $\cS_1$
and all implementations $\cI_1$, $\cI_2$, that
\begin{equation}
  \label{eq:quotient}
  \cI_1\in \sem{ \cS_1} \text{ and } \cI_2\in \sem{ \cS\by{} \cS_1}
  \text{ imply } \cI_1\para{} \cI_2\in \sem \cS.
\end{equation}
Furthermore, $\cS\by{} \cS_1$ is to be as permissive as possible.

We can again obtain such a quotient operator for $\nu$-calculus from the
one for AA introduced in~\cite{DBLP:conf/concur/BenesDFKL13}.  Hence we
define, for $\nu$-calculus expressions $\cN_1$, $\cN_2$ in normal form,
$\cN_1\by{} \cN_2= \bh( \hb( \cN_1)\by A \hb( \cN_2))$, where $\by A$ is
AA quotient.  We recall the construction of $\by A$
from~\cite{DBLP:conf/concur/BenesDFKL13}:

Let $\cA_1=( S_1, S_1^0, \Tran_1)$, $\cA_2=( S_2, S_2^0, \Tran_2)$ be AA
and define $\cA_1\by A \cA_2=( S,\{ s^0\}, \Tran)$, with $S= 2^{
  S_1\times S_2}$, $s^0=\{( s_1^0, s_2^0)\mid s_1^0\in S_1^0, s_2^0\in
S_2^0\}$, and $\Tran$ given as follows:

Let $\Tran( \emptyset)= 2^{ \Sigma\times\{ \emptyset\}}$.  For $s=\{(
s_1^1, s_2^1),\dots,( s_1^n, s_2^n)\}\in S$, say that $a\in \Sigma$ is
\emph{permissible from $s$} if it holds for all $i= 1,\dots, n$ that
there is $M_1\in \Tran_1( s_1^i)$ and $t_1\in S_1$ for which $( a,
t_1)\in M_1$, or else there is no $M_2\in \Tran_2( s_2^i)$ and no
$t_2\in S_2$ for which $( a, t_2)\in M_2$.

For $a$ permissible from $s$ and $i\in\{ 1,\dots, n\}$, let $\{ t_2^{ i,
  1},\dots, t_2^{ i, m_i}\}=\{ t_2\in S_2\mid \exists M_2\in \Tran_2(
s_2^i):( a, t_2)\in M_2\}$ be an enumeration of the possible states in
$S_2$ after an $a$-transition and define $\postra[ a]{ s}=\big\{\{(
t_1^{ i, j}, t_2^{ i, j})\mid i= 1,\dots, n, j= 1,\dots, m_i\}\bigmid
\forall i: \forall j: \exists M_1\in \Tran_1( s_1^i):( a, t_1^{ i,
  j})\in M_1\big\}$, the set of all sets of possible assignments of
next-$a$ states from $s_1^i$ to next-$a$ states from $s_2^i$.

Now let $\postra{ s}=\{( a, t)\mid t\in \postra[ a]{ s}, a\text{
  admissible from } s\}$ and define $\Tran( s)=\{ M\subseteq \postra{
  s}\mid \forall i= 1,\dots, n: \forall M_2\in \Tran_2( s_2^i):
M\triangleright M_2\in \Tran_1( s_1^i)\}$.  Here $\triangleright$ is the
composition-projection operator defined by $M\triangleright M_2=\{( a,
t\triangleright t_2)\mid( a, t)\in M,( a, t_2)\in M_2\}$ and
$t\triangleright t_2=\{( t_1^1, t_2^1),\dots,( t_1^k,
t_2^k)\}\triangleright t_2^i= t_1^i$ (note that by construction, there
is precisely one pair in $t$ whose second component is $t_2^i$).

\begin{theorem}
  \label{th:quotient}
  For all $\nu$-calculus expressions $\cN$, $\cN_1$, $\cN_2$ in normal
  form, $\cN_2\mr \cN \by{} \cN_1$ iff $\cN_1\para{} \cN_2\mr \cN$.
\end{theorem}

\begin{proof}
  From the analogous property for AA~\cite{DBLP:conf/concur/BenesDFKL13}
  and Theorems~\ref{th:trans-modref}, \ref{th:dmts->nu}
  and~\ref{th:nu->dmts}. \qed
\end{proof}

As a corollary, we get~\eqref{eq:quotient}: If $\cI_2\in \sem{ \cN\by{}
  \cN_1}$, \ie~$\cI_2\mr \cN\by{} \cN_1$, then $\cN_1\para{} \cI_2\mr
\cN$, which using $\cI_1\mr \cN_1$ and Theorem~\ref{th:indimp} implies
$\cI_1\para{} \cI_2\mr \cN_1\para{} \cI_2\mr \cN$.  The reverse
implication in Theorem~\ref{th:quotient} implies that $\cN\by{} \cN_1$
is as permissive as possible.

\begin{theorem}
  \label{th:resilat}
  With operations $\land$, $\lor$, $\para{}$ and $\by{}$, the class of
  $\nu$-calculus expressions forms a commutative residuated lattice up
  to $\mreq$.
\end{theorem}

The unit of $\para{}$ (up to $\mreq$) is the $\nu$-calculus expression
corresponding to the LTS $\cunit=(\{ u\},\{ u\},\{( u, a, u)\mid a\in
\Sigma\})$.  We refer to~\cite{journal/ijac/HartRT02} for a good
reference on commutative residuated lattices.

\section{Conclusion and Further Work}

Using new translations between the modal $\nu$-calculus and DMTS, we
have exposed a structural equivalence between these two specification
formalisms.  This means that both types of specifications can be freely
mixed; there is no more any need to decide, whether due to personal
preference or for technical reasons, between one and the other.  Of
course, the modal $\nu$-calculus can only express safety properties; for
more expressivity, one has to turn to more expressive logics, and no
behavioral analogue to these stronger logics is known (neither is it
likely to exist, we believe).

Our constructions of composition and quotient for the modal
$\nu$-calculus expect (and return) $\nu$-calculus expressions in normal
form, and it is an interesting question whether they can be defined for
general $\nu$-calculus expressions.  (For disjunction and conjunction
this is of course trivial.)  Larsen's~\cite{DBLP:conf/concur/Larsen90}
has composition and quotient operators for Hennessy-Milner logic
(restricted to ``deterministic context systems''), but we know of no
extension (other than ours) to more general logics.

We also note that our hybrid modal logic appears related to the
\emph{Boolean equation
  systems}~\cite{thesis/Mader97,DBLP:conf/cav/Larsen92} which are used
in some $\mu$-calculus model checking algorithms.  The precise relation
between the modal $\nu$-calculus, our $\cL$-expressions and Boolean
equation systems should be worked out.  Similarly, acceptance automata
bear some similarity to the \emph{modal automata}
of~\cite{book/BradfieldS06}.

Lastly, we should note that we have
in~\cite{DBLP:conf/csr/BauerFLT12,DBLP:conf/mfcs/BauerFJLLT11}
introduced \emph{quantitative} specification theories for weighted modal
transition systems.  These are well-suited for specification and
analysis of systems with quantitative information, in that they replace
the standard Boolean notion of refinement with a robust distance-based
notion.  We are working on an extension of these quantitative formalisms
to DMTS, and hence to the modal $\nu$-calculus, which should relate our
work to other approaches at quantitative model checking such as
\eg~\cite{DBLP:conf/icalp/AlfaroHM03,DBLP:conf/concur/Alfaro03,DBLP:conf/concur/GeblerF12}.

\bibliographystyle{myabbrv}
\bibliography{mybib}

\begin{thebibliography}{10}

\bibitem{books/AcetoILS07}
L.~Aceto, A.~Ing{\'o}lfsd{\'o}ttir, K.~G. Larsen, and J.~Srba.
\newblock {\em Reactive Systems}.
\newblock {Cambridge Univ. Press}, 2007.

\bibitem{DBLP:conf/fase/BauerDHLLNW12}
S.~S. Bauer, A.~David, R.~Hennicker, K.~G. Larsen, A.~Legay, U.~Nyman, and
  A.~W{\k a}sowski.
\newblock Moving from specifications to contracts in component-based design.
\newblock In {\em FASE}, vol. 7212 of {\em {LNCS}}. Springer, 2012.

\bibitem{DBLP:conf/mfcs/BauerFJLLT11}
S.~S. Bauer, U.~Fahrenberg, L.~Juhl, K.~G. Larsen, A.~Legay, and C.~Thrane.
\newblock Quantitative refinement for weighted modal transition systems.
\newblock In {\em MFCS}, vol. 6907 of {\em {LNCS}}. {Springer}, 2011.

\bibitem{DBLP:conf/csr/BauerFLT12}
S.~S. Bauer, U.~Fahrenberg, A.~Legay, and C.~Thrane.
\newblock General quantitative specification theories with modalities.
\newblock In {\em CSR}, vol. 7353 of {\em {LNCS}}. {Springer}, 2012.

\bibitem{DBLP:conf/atva/BenesCK11}
N.~Bene{\v s}, I.~{\v C}ern{\'a}, and J.~K{\v r}et\'{\i}nsk{\'y}.
\newblock Modal transition systems: Composition and {LTL} model checking.
\newblock In  \cite{DBLP:conf/atva/2011}.

\bibitem{DBLP:conf/concur/BenesDFKL13}
N.~Bene{\v s}, B.~Delahaye, U.~Fahrenberg, J.~K{\v r}et\'{\i}nsk{\'y}, and
  A.~Legay.
\newblock {H}ennessy-{M}ilner logic with greatest fixed points.
\newblock In {\em CONCUR}, vol. 8052 of {\em {LNCS}}. {Springer}, 2013.

\bibitem{DBLP:conf/atva/BenesKLMS11}
N.~Bene{\v s}, J.~K{\v r}et\'{\i}nsk{\'y}, K.~G. Larsen, M.~H. M{\o}ller, and
  J.~Srba.
\newblock Parametric modal transition systems.
\newblock In  \cite{DBLP:conf/atva/2011}.

\bibitem{DBLP:journals/tcs/BenesKLS09}
N.~Bene{\v s}, J.~K{\v r}et\'{\i}nsk{\'y}, K.~G. Larsen, and J.~Srba.
\newblock On determinism in modal transition systems.
\newblock {\em {Th. Comp. Sci.}}, 410(41):4026--4043, 2009.

\bibitem{DBLP:journals/iandc/BenesKLS12}
N.~Bene{\v s}, J.~K{\v r}et{\'\i}nsk{\'y}, K.~G. Larsen, and J.~Srba.
\newblock {EXPTIME}-completeness of thorough refinement on modal transition
  systems.
\newblock {\em Inf. Comp.}, 218:54--68, 2012.

\bibitem{DBLP:journals/igpl/Blackburn00}
P.~Blackburn.
\newblock Representation, reasoning, and relational structures: a hybrid logic
  manifesto.
\newblock {\em Log. J. IGPL}, 8(3):339--365, 2000.

\bibitem{DBLP:journals/tcs/BoudolL92}
G.~Boudol and K.~G. Larsen.
\newblock Graphical versus logical specifications.
\newblock {\em {Th. Comp. Sci.}}, 106(1):3--20, 1992.

\bibitem{book/BradfieldS06}
J.~Bradfield and C.~Stirling.
\newblock Modal mu-calculi.
\newblock In {\em The Handbook of Modal Logic}. Elsevier, 2006.

\bibitem{DBLP:conf/atva/2011}
T.~Bultan and P.-A. Hsiung, eds.
\newblock {\em Automated Technology for Verification and Analysis, 9th Int.
  Symp., ATVA 2011}, vol. 6996 of {\em {LNCS}}. {Springer}, 2011.

\bibitem{DBLP:journals/iandc/CairesC03}
L.~Caires and L.~Cardelli.
\newblock A spatial logic for concurrency.
\newblock {\em Inf. Comp.}, 186(2), 2003.

\bibitem{DBLP:conf/lop/ClarkeE81}
E.~M. Clarke and E.~A. Emerson.
\newblock Design and synthesis of synchronization skeletons using
  branching-time temporal logic.
\newblock In {\em Logic of Programs}, vol. 131 of {\em {LNCS}}. Springer, 1981.

\bibitem{DBLP:conf/concur/Alfaro03}
L.~de~Alfaro.
\newblock Quantitative verification and control via the mu-calculus.
\newblock In {\em CONCUR}, vol. 2761 of {\em {LNCS}}. {Springer}, 2003.

\bibitem{DBLP:conf/icalp/AlfaroHM03}
L.~de~Alfaro, T.~A. Henzinger, and R.~Majumdar.
\newblock Discounting the future in systems theory.
\newblock In {\em ICALP}, vol. 2719 of {\em {LNCS}}. {Springer}, 2003.

\bibitem{DBLP:conf/concur/GeblerF12}
D.~Gebler and W.~Fokkink.
\newblock Compositionality of probabilistic {H}ennessy-{M}ilner logic through
  structural operational semantics.
\newblock In {\em CONCUR}, vol. 7454 of {\em {LNCS}}. {Springer}, 2012.

\bibitem{journal/ijac/HartRT02}
J.~B. Hart, L.~Rafter, and C.~Tsinakis.
\newblock The structure of commutative residuated lattices.
\newblock {\em Internat. J. Algebra Comput.}, 12(4):509--524, 2002.

\bibitem{DBLP:journals/jacm/Hennessy85}
M.~Hennessy.
\newblock Acceptance trees.
\newblock {\em J. ACM}, 32(4):896--928, 1985.

\bibitem{DBLP:journals/tcs/Kozen83}
D.~Kozen.
\newblock Results on the propositional $\mu$-calculus.
\newblock {\em {Th. Comp. Sci.}}, 27, 1983.

\bibitem{DBLP:conf/avmfss/Larsen89}
K.~G. Larsen.
\newblock Modal specifications.
\newblock In {\em Automatic Verification Methods for Finite State Systems},
  vol. 407 of {\em {LNCS}}. {Springer}, 1989.

\bibitem{DBLP:conf/concur/Larsen90}
K.~G. Larsen.
\newblock Ideal specification formalism = expressivity + compositionality +
  decidability + testability + ...
\newblock In {\em CONCUR}, vol. 458 of {\em {LNCS}}. {Springer}, 1990.

\bibitem{DBLP:journals/tcs/Larsen90}
K.~G. Larsen.
\newblock Proof systems for satisfiability in {H}ennessy-{M}ilner logic with
  recursion.
\newblock {\em {Th. Comp. Sci.}}, 72(2{\&}3):265--288, 1990.

\bibitem{DBLP:conf/cav/Larsen92}
K.~G. Larsen.
\newblock Efficient local correctness checking.
\newblock In {\em CAV}, vol. 663 of {\em {LNCS}}. {Springer}, 1992.

\bibitem{DBLP:conf/lics/LarsenX90}
K.~G. Larsen and L.~Xinxin.
\newblock Equation solving using modal transition systems.
\newblock In {\em LICS}. IEEE Computer Society, 1990.

\bibitem{thesis/Mader97}
A.~Mader.
\newblock {\em Verification of Modal Properties Using Boolean Equation
  Systems}.
\newblock PhD thesis, Technische Universit{\"a}t M{\"u}nchen, 1997.

\bibitem{DBLP:conf/csl/OHearnRY01}
P.~W. O'Hearn, J.~C. Reynolds, and H.~Yang.
\newblock Local reasoning about programs that alter data structures.
\newblock In {\em CSL}, vol. 2142 of {\em {LNCS}}. Springer, 2001.

\bibitem{book/Prior68}
A.~N. Prior.
\newblock {\em Papers on Time and Tense}.
\newblock Oxford: Clarendon Press, 1968.

\bibitem{report/irisa/Raclet07}
J.-B. Raclet.
\newblock Residual for component specifications.
\newblock Publication interne 1843, IRISA, Rennes, 2007.

\bibitem{DBLP:journals/entcs/Raclet08}
J.-B. Raclet.
\newblock Residual for component specifications.
\newblock {\em Electr. Notes Theor. Comput. Sci.}, 215:93--110, 2008.

\bibitem{DBLP:conf/lics/Reynolds02}
J.~C. Reynolds.
\newblock Separation logic: A logic for shared mutable data structures.
\newblock In {\em LICS}. IEEE Computer Society, 2002.

\bibitem{unpub/ScottB69}
D.~Scott and J.~W. de~Bakker.
\newblock A theory of programs.
\newblock Unpublished manuscript, IBM, Vienna, 1969.

\end{thebibliography}

\newpage

\section*{Appendix: Extra Lemmas and Proofs}

\begin{lemma}
  \label{le:dmtstobfsspecial}
  Let $\cD=( S, S^0, \omay, \omust)$ be a DMTS and $s\in S$.  For all
  $M_1, M_2\in \Tran( s)$ and all $M\subseteq \Sigma\times S$ with
  $M_1\subseteq M\subseteq M_1\cup M_2$, also $M\in \Tran( s)$.
\end{lemma}

\begin{proof}
  For $i=1,2$, since $M_i \in \Tran(s)$, we know that
  \begin{itemize}
  \item for all $(a,t) \in M_i$, $(s,a,t) \in \omay$, and
  \item for all $(s,N) \in \omust$, there is $(a,t) \in M_i \cap N$. 
  \end{itemize}

  Now as $M \subseteq M_1 \cup M_2$, it directly follows that for all
  $(a,t) \in M$, we have $(s,a,t) \in \omay$. Moreover, since $M_1
  \subseteq M$, we also have that for all $(s,N) \in \omust$, there
  exists $(a,t) \in M \cap N$.  As a consequence, $M \in \Tran(s)$. \qed
\end{proof}

\begin{proof}[of Lemma~\ref{le:bfstodmtsblowup}]
  Let $\Sigma=\{ a_1,\dots, a_n\}$ and $\cA=(\{ s^0\},\{ s^0\}, \Tran)$
  the AA with $\Tran( s^0)=\{ M\subseteq \Sigma\times\{ s^0\}\mid
  \exists k:| M|= 2k\}$ the transition constraint containing all
  disjunctive choices of even cardinality.  Let $\cD=( T, T^0, \omay,
  \omust)$ be a DMTS with $cD\treq \cA$; we claim that $\cD$ must have
  at least $2^{ n- 1}$ initial states.

  Assume, for the purpose of contradiction, that $T^0=\{ t^0_1,\dots,
  t^0_m\}$ with $m< 2^{ n- 1}$.  We must have $\bigcup_{ i= 1}^m
  \Tran_T( t^0_i)=\{ M\subseteq \Sigma\times T\mid \exists k:| M|=
  2k\}$, so that there is an index $j\in\{ 1,\dots, m\}$ for which
  $\Tran_T( t^0_j)=\{ M_1, M_2\}$ contains two different disjunctive
  choices from $\Tran_S( s^0)$.  By Lemma~\ref{le:dmtstobfsspecial},
  also $M\in \Tran_T( t^0_j)$ for any $M$ with $M_1\subseteq M\subseteq
  M_1\cup M_2$.  But $M_1\cup M_2$ has greater cardinality than $M_1$,
  so that there will be an $M\in \Tran_T( t^0_j)$ with odd cardinality. \qed
\end{proof}

\begin{proof}[of Theorem~\ref{th:trans-modref}]
  The first two equivalences in the theorem follow directly from the
  definitions.  Indeed, for the translation from $\cL$-expressions to
  AA, we have $\lsem{ \Phi( x)}= \Tran( x)$ by definition, hence $M\in
  \lsem{ \Phi( x)}$ iff $M\in \Tran( x)$.  For the other translation, we
  compute
  \begin{align*}
    \lsem{ \Phi( x)} &= \lsem{ \biglor_{ M\in \Tran( x)} \big(
      \bigland_{( a, t)\in M} \langle a\rangle t\land \bigland_{( b,
        u)\notin M} \neg \langle b\rangle u\big)} \\
    &= \bigcup_{ M\in \Tran( x)} \big( \bigcap_{( a, t)\in M} \{ M'\mid(
    a, t)\in M'\} \cap \bigcap_{( b, u)\notin M} \{ M'\mid( b, u)\notin
    M'\}\big) \\
    &= \bigcup_{ M\in \Tran( x)} \big( \{ M'\mid \forall( a, t)\in M:(
    a, t)\in M'\} \\[-3ex]
    & \hspace*{14em} \cap\{ M'\mid \forall( b, u)\notin M:( b, u)\notin
    M'\}\big) \\[1ex]
    &= \bigcup_{ M\in \Tran( x)} \big( \{ M'\mid M\subseteq M'\}\cap\{
    M'\mid M'\subseteq M\}\big) \\
    &= \bigcup_{ M\in \Tran( x)} M= \Tran( x).
  \end{align*}

  \smallskip \noindent \underline{$\cD_1\mr \cD_2$ implies $\db(
    \cD_1)\mr \db( \cD_2)$:}

  Let $\cD_1=( S_1, S^0_1, \omay_1, \omust_1)$, $\cD_2=( S_2, S^0_2,
  \omay_2, \omust_2)$ be DMTS and assume $\cD_1\mr \cD_2$.  Then we have
  a modal refinement relation (in the DMTS sense) $R\subseteq S_1\times
  S_2$.  Now let $( s_1, s_2)\in R$ and $M_1\in \Tran_1( s_1)$, and
  define
  \begin{equation*}
    M_2=\{( a, t_2)\mid s_2\may{a}_2 t_2, \exists( a, t_1)\in M_1:( t_1,
    t_2)\in R\}.
  \end{equation*}

  The condition
  \begin{equation*}
    \forall( a, t_2)\in M_2: \exists( a, t_1)\in M_1:( t_1, t_2)\in
    R
  \end{equation*}
  in the definition of AA refinement is satisfied by construction.  For
  the inverse condition, let $( a, t_1)\in M_1$, then $s_1\may{a}_1
  t_1$, so by DMTS refinement, there is $t_2\in S_2$ with $s_2\may{a}_2
  t_2$ and $( t_1, t_2)\in R$, whence $( a, t_2)\in M_2$ by construction.

  We are left with showing that $M_2\in \Tran_2( s_2)$.  First we notice
  that by construction, indeed $s_2\may{a}_2 t_2$ for all $( a, t_2)\in
  M_2$.  Now let $s_2\must{} N_2$; we need to show that $N_2\cap M_2\ne
  \emptyset$.

  By DMTS refinement, we have $s_1\must{} N_1$ such that $\forall( a,
  t_1)\in N_1: \exists( a, t_2)\in N_2:( t_1, t_2)\in R$.  We know that
  $N_1\cap M_1\ne \emptyset$, so let $( a, t_1)\in N_1\cap M_1$.  Then
  there also is $( a, t_2)\in N_2$ with $( t_1, t_2)\in R$.  But $( a,
  t_2)\in N_2$ implies $s_2\may{a}_2 t_2$, hence $( a, t_2)\in M_2$.

  \smallskip \noindent \underline{$\db( \cD_1)\mr \db( \cD_2)$ implies
    $\cD_1\mr \cD_2$:}

  Let $\cD_1=( S_1, S^0_1, \omay_1, \omust_1)$, $\cD_2=( S_2, S^0_2,
  \omay_2, \omust_2)$ be DMTS and assume $\db( \cD_1)\mr \db( \cD_2)$.
  Then we have a modal refinement relation (in the AA sense) $R\subseteq
  S_1\times S_2$.  Let $( s_1, s_2)\in R$.

  Let $s_1\may{a}_1 t_1$, then we cannot have $s_1\must{} \emptyset$.
  Let $M_1=\{( a, t_1)\}\cup \bigcup\{ N_1\mid s_1\must{} N_1\}$, then
  $M_1\in \Tran_1( s_1)$ by construction.  This implies that there is
  $M_2\in \Tran_2( s_2)$ and $( a, t_2)\in M_2$ with $( t_1, t_2)\in R$,
  but then also $s_2\may{a} t_2$ as was to be shown.

  Let $s_2\must{} N_2$ and assume, for the sake of contradiction, that
  there is no $s_1\must{} N_1$ for which $\forall( a, t_1)\in N_1:
  \exists( a, t_2)\in N_2:( t_1, t_2)\in R$ holds.  Then for each
  $s_1\must{} N_1$, there is an element $( a_{ N_1}, t_{ N_1})\in N_1$
  for which there is no $( a_{ N_1}, t_2)\in N_2$ with $( t_{ N_1},
  t_2)\in R$.

  Let $M_1=\{( a_{ N_1}, t_{ N_1})\mid s_1\must{} N_1\}$, then $M_1\in
  \Tran_1( s_1)$ by construction.  Hence we have $M_2\in \Tran_2( s_2)$
  satisfying the conditions in the definition of AA refinement.  By
  construction of $\Tran_2( s_2)$, $N_2\cap M_2\ne \emptyset$, so let $(
  a, t_2)\in N_2\cap M_2$.  Then there exists $( a, t_1)\in M_1$ for
  which $( t_1, t_2)\in R$, in contradiction to the definition of
  $M_1$.

  \smallskip \noindent \underline{$\cA_1\mr \cA_2$ implies $\bd( \cA_1)\mr \bd(
    \cA_2)$:}

  Let $\cA_1=( S_1, S^0_1, \Tran_1)$, $\cA_2=( S_2, S^0_2, \Tran_2)$ be
  AA, with DMTS translations $( D_1, D^0_1, \omust_1, \omay_1)$, $( D_2,
  D^0_2, \omust_2, \omay_2)$, and assume $\cA_1\mr \cA_2$.  Then we have
  a modal refinement relation (in the AA sense) $R\subseteq S_1\times
  S_2$.  Define $R'\subseteq D_1\times D_2$ by
  \begin{multline*}
    R'= \{( M_1, M_2)\mid \exists( s_1, s_2)\in R: M_1\in \Tran_1( s_1),
    M_2\in \Tran( s_2), \\
    \begin{aligned}
      & \forall( a, t_1)\in M_1: \exists( a, t_2)\in M_2:( t_1, t_2)\in
      R, \\
      & \forall( a, t_2)\in M_2: \exists( a, t_1)\in M_1:( t_1, t_2)\in
      R \}.
    \end{aligned}
  \end{multline*}
  We show that $R'$ is a modal refinement in the DMTS sense.  Let $(
  M_1, M_2)\in R'$.

  Let $M_2\must{}_2 N_2$.  By construction of $\omust$, there is $( a,
  t_2)\in M_2$ such that $N_2=\{( a, M_2')\mid M_2'\in \Tran_2( t_2)\}$.
  Then $( M_1, M_2)\in R'$ implies that there must be $( a, t_1)\in M_1$
  for which $( t_1, t_2)\in R$, and we can define $N_1=\{( a, M_1')\mid
  M_1'\in \Tran_1( t_1)\}$, whence $M_1\must{}_1 N_1$.

  We show that $\forall( a, M_1')\in N_1: \exists( a, M_2')\in N_2:(
  M_1', M_2')\in R'$: Let $( a, M_1')\in N_1$, then $M_1'\in \Tran_1(
  t_1)$.  From $( t_1, t_2)\in R$ we hence get $M_2'\in \Tran_2( t_2)$,
  and then $( a, M_2')\in N_2$ by construction of $N_2$ and $( M_1',
  M_2')\in R'$ due to the conditions of AA refinement (applied to $(
  t_1, t_2)\in R$).

  Let $M_1\may{a}_1 M_1'$, then we have $M_1\must{}_1 N_1$ for which $(
  a, M_1')\in N_1$ by construction of $\omay$.  This in turn implies
  that there must be $( a, t_1)\in M_1$ such that $N_1=\{( a, M_1'')\mid
  M_1''\in \Tran_1( t_1)\}$, and then by $( M_1, M_2)\in R'$, we get $(
  a, t_2)\in M_2$ for which $( t_1, t_2)\in R$.  Let $N_2=\{( a,
  M_2')\mid M_2'\in \Tran_2( t_2)\}$, then $M_2\must{}_2 N_2$ and hence
  $M_2\may{a}_2 M_2'$ for all $( a, M_2')\in N_2$.  On the other hand,
  the argument in the previous paragraph shows that there is $( a,
  M_2')\in N_2$ for which $( M_1', M_2')\in R'$.

  We miss to show that $R'$ is initialized.  Let $M_1^0\in D_1^0$, then
  we have $s_1^0\in S_1^0$ with $M_1^0\in \Tran_1( s_1^0)$.  As $R$ is
  initialized, this entails that there is $s_2^0\in S_2^0$ with $(
  s_1^0, s_2^0)\in R$, which gives us $M_2^0\in \Tran_2( s_2^0)$ which
  satisfies the AA refinement conditions, whence $( M_1^0, M_2^0)\in R'$.

  \smallskip \noindent \underline{$\bd( \cA_1)\mr \bd( \cA_2)$ implies $\cA_1\mr
    \cA_2$:}

  Let $\cA_1=( S_1, S^0_1, \Tran_1)$, $\cA_2=( S_2, S^0_2, \Tran_2)$ be
  AA, with DMTS translations $( D_1, D^0_1, \omust_1, \omay_1)$, $( D_2,
  D^0_2, \omust_2, \omay_2)$, and assume $\bd( \cA_1)\mr \bd( \cA_2)$.
  Then we have a modal refinement relation (in the DMTS sense)
  $R\subseteq D_1\times D_2$.  Define $R'\subseteq S_1\times S_2$ by
  \begin{equation*}
    R'=\{( s_1, s_2)\mid \forall M_1\in \Tran_1( s_1): \exists M_2\in
    \Tran_2( s_2):( M_1, M_2)\in R\};
  \end{equation*}
  we will show that $R'$ is an AA modal refinement.

  Let $( s_1, s_2)\in R'$ and $M_1\in \Tran_1( s_1)$, then by
  construction of $R'$, we have $M_2\in \Tran_2( s_2)$ with $( M_1,
  M_2)\in R$.

  Let $( a, t_2)\in M_2$ and define $N_2=\{( a, M_2')\mid M_2'\in
  \Tran_2( t_2)\}$, then $M_2\must{}_2 N_2$.  Now $( M_1, M_2)\in R$
  implies that there must be $M_1\must{}_1 N_1$ satisfying $\forall( a,
  M_1')\in N_1: \exists( a, M_2')\in N_2:( M_1', M_2')\in R$.  We have
  $( a, t_1)\in M_1$ such that $N_1=\{( a, M_1')\mid M_1'\in \Tran_1(
  t_1)\}$; we only miss to show that $( t_1, t_2)\in R'$.  Let $M_1'\in
  \Tran_1( t_1)$, then $( a, M_1')\in N_1$, hence there is $( a,
  M_2')\in N_2$ with $( M_1', M_2')\in R$, but $( a, M_2')\in N_2$ also
  entails $M_2'\in \Tran_2( t_2)$.

  Let $( a, t_1)\in M_1$ and define $N_1=\{( a, M_1')\mid M_1'\in
  \Tran_1( t_1)\}$, then $M_1\must{}_1 N_1$.  Now let $( a, M_1')\in
  N_1$, then $M_1\may{a}_1 M_1'$, hence we have $M_2\may{a}_2 M_2'$ for
  some $( M_1', M_2')\in R$ by modal refinement.  By construction of
  $\omay$, this implies that there is $M_2\must{}_2 N_2$ with $( a,
  M_2')\in N_2$, and we have $( a, t_2)\in M_2$ for which $N_2=\{( a,
  M_2'')\mid M_2''\in \Tran_2( t_2)\}$.  Now if $M_1''\in \Tran_1(
  t_1)$, then $( a, M_1'')\in N_1$, hence there is $( a, M_2'')\in N_2$
  with $( M_1'', M_2'')\in R$, but $( a, M_2'')\in N_2$ also gives
  $M_2''\in \Tran_2( t_2)$.

  We miss to show that $R'$ is initialized.  Let $s^0_1\in S^0_1$, then
  $\Tran_1( s^0_1)\ne \emptyset$, hence there is $M^0_1\in \Tran_1(
  s^0_1)$.  As $R$ is initialized, this gets us $M^0_2\in D_2$ with $(
  M^0_1, M^0_2)\in R$, but $M^0_2\in \Tran_2( s^0_2)$ for some $s^0_2\in
  S^0_2$, and then $( s^0_1, s^0_2)\in R'$. \qed
\end{proof}

\begin{proof}[of Lemma~\ref{le:hmlnormal}]
  It is shown in~\cite{DBLP:journals/tcs/BoudolL92} that any
  Hennessy-Milner formula is equivalent to one in so-called \emph{strong
    normal form}, \ie~of the form $\biglor_{ i\in I}( \bigland_{ j\in
    J_i} \langle a_{ ij}\rangle \phi_{ ij}\land \bigland_{ a\in \Sigma}[
  a] \psi_{ i, a})$ for HML formulas $\phi_{ ij}$, $\psi_{ i, a}$ which
  are also in strong normal form.  Now we can replace the $\phi_{ ij}$,
  $\psi_{ i, a}$ by (new) variables $x_{ ij}$, $y_{ i, a}$ and add
  declarations $\Delta_2( x_{ ij})= \phi_{ ij}$, $\Delta_2( y_{ i, a})=
  \psi_{ i, a}$ to arrive at an expression in which all formulae are of
  the form $\Delta_2( x)= \biglor_{ i\in I}( \bigland_{ j\in J_i} \langle
  a_{ ij}\rangle x_{ ij}\land \bigland_{ a\in \Sigma}[ a] x_{ i, a})$.

  Now for each such formula, replace (recursively) $x$ by new variables
  $\{ \tilde x^i\mid i\in I\}$ and set $\Delta_2( \tilde x^i)= \bigland_{
    j\in J_i}\langle a_{ ij}\rangle\big( \biglor_k \tilde
  x_{ij}^k\big)\land \bigland_{ a\in \Sigma}[ a]\big( \biglor_k \tilde
  x_{ i, a}^k\big)$.  Using initial variables $X_2^0=\{ \tilde x^i\mid
  x\in X_1^0\}$, the so-constructed $\nu$-calculus expression is
  equivalent to the original one. \qed
\end{proof}

\begin{proof}[of Theorem~\ref{th:logops}]
  This follows directly from the fact, easy to prove because of the
  purely syntactic translation, that $\hd( \cN_1\lor \cN_2)= \hd(
  \cN_1)\lor \hd( \cN_2)$ and similarly for conjunction, where the
  operations on the right-hand side are the ones defined for DMTS
  in~\cite{DBLP:conf/concur/BenesDFKL13}.  Given this, the theorem
  follows from similar properties for the DMTS
  operations~\cite{DBLP:conf/concur/BenesDFKL13}. \qed
\end{proof}

\begin{proof}[of Theorem~\ref{th:distlat}]
  Only distributivity remains to be verified.  Let $\cN_i=( X_i, X_i^0,
  \Delta_i)$, for $i= 1, 2, 3$, be $\nu$-calculus expressions in normal
  form.  The set of variables of both $\cN_1\land( \cN_2\lor \cN_3)$ and
  $( \cN_1\land \cN_2)\lor( \cN_1\land \cN_3)$ is $X_1\times( X_2\times
  X_3)$, and one easily sees that the identity relation is a two-sided
  modal refinement.  Things are similar for the other distributive law. \qed
\end{proof}

\begin{proof}[of Theorem~\ref{th:resilat}]
  We have already seen that the class of $\nu$-calculus expressions
  forms a lattice, up to $\mreq$, under $\land$ and $\lor$, and by
  Theorem~\ref{th:quotient}, $\by{}$ is the residual, up to $\mreq$, of
  $\para{}$.  We only miss to show that $\cunit$ is indeed the unit of
  $\para{}$; all other properties (such as distributivity of $\para{}$
  over $\lor$ or $\cN\para{} \bot\mreq \bot$) follow.

  We show that $\cA\para A \cunit\mreq \cA$ for all AA $\cA$; the
  analogous property for $\nu$-calculus expressions follows from the
  translations.  Let $\cA=( S, S^0, \Tran)$ be an AA and define $R=\{((
  s, u), s)\mid s\in S\}$; we show that $R$ is a two-sided modal
  refinement.  Let $(( s, u), s)\in R$ and $M\in \Tran( s, u)$, then
  there must be $M_1\in \Tran( s)$ for which $M= M_1\|( \Sigma\times\{
  u\})$.  Thus $M_1=\{( a, t)\mid( a,( t, u))\in M\}$.  Then any element
  of $M$ has a corresponding one in $M_1$, and vice versa, and their
  states are related by $R$.  For the other direction, let $M_1\in
  \Tran( s)$, then $M= M_1\|( \Sigma\times\{ u\})=\{( a,( t, u))\mid( a,
  t)\in M_1\}\in \Tran( s, u)$, and the same argument applies. \qed
\end{proof}

\end{document}